\documentclass[a4paper, 11pt, onecolumn, unpublished]{quantumarticle}
\pdfoutput=1

\usepackage{setspace}
\usepackage{adjustbox}
\usepackage{etoolbox}
\usepackage{layout}

\usepackage[english]{babel}
\usepackage[utf8]{inputenc}
\usepackage[T1]{fontenc}
\usepackage[scaled=0.9]{beramono}

\usepackage[mathscr]{euscript}
\usepackage{amsthm}
\usepackage{mathtools}

\usepackage{fancyhdr}
\usepackage[bottom]{footmisc}

\usepackage{tikzit}
\usepackage{graphicx}
\usepackage{subcaption}
\usepackage{wrapfig}
\usepackage{float}

\usepackage{booktabs}
\usepackage{multirow}
\usepackage{array}
\usepackage{sidecap}

\usepackage{algorithm}
\usepackage{algpseudocode}
\usepackage{marginnote}

\usepackage[numbers,sort&compress]{natbib}

\tikzstyle{box}=[shape=rectangle, text height=1.5ex, text depth=0.25ex, yshift=0.5mm, fill=white, draw=black, minimum height=5mm, yshift=-0.5mm, minimum width=5mm, font={\small}]
\tikzstyle{Z dot}=[inner sep=0mm, minimum size=2mm, shape=circle, draw=black, fill={rgb,255: red,221; green,255; blue,221}]
\tikzstyle{Z phase dot}=[minimum size=5mm, font={\footnotesize\boldmath}, shape=rectangle, rounded corners=2mm, inner sep=0.2mm, outer sep=-2mm, scale=0.8, tikzit shape=circle, draw=black, fill={rgb,255: red,221; green,255; blue,221}, tikzit draw=blue]
\tikzstyle{X dot}=[Z dot, shape=circle, draw=black, fill={rgb,255: red,255; green,136; blue,136}]
\tikzstyle{X phase dot}=[Z phase dot, tikzit shape=circle, tikzit draw=blue, fill={rgb,255: red,255; green,136; blue,136}, font={\footnotesize\boldmath}]
\tikzstyle{hadamard}=[fill=yellow, draw=black, shape=rectangle, inner sep=0.6mm, minimum height=1.5mm, minimum width=1.5mm]
\tikzstyle{vertex}=[inner sep=0mm, minimum size=1mm, shape=circle, draw=black, fill=black]
\tikzstyle{vertex set}=[inner sep=0mm, minimum size=1mm, shape=circle, draw=black, fill=white, font={\footnotesize\boldmath}]

\tikzstyle{hadamard edge}=[-, dashed, dash pattern=on 2pt off 0.5pt, thick, draw={rgb,255: red,68; green,136; blue,255}]
\tikzstyle{brace edge}=[-, tikzit draw=blue, decorate, decoration={brace,amplitude=1mm,raise=-1mm}]
\tikzstyle{diredge}=[->]
\tikzstyle{dir hadamard edge}=[->, dashed, dash pattern=on 2pt off 0.5pt, thick, draw={rgb,255: red,68; green,136; blue,255}]

\input{zx.tikzdefs}
\usepackage[capitalise,noabbrev]{cleveref}

\renewcommand{\footnoterule}{%
    \kern-2pt
    \hrule width 0.42\columnwidth
    \kern 2.6pt}

\newcounter{fnmark}
\newcounter{fntext}
\newcommand{\fnmark}[2]{\hyperref[#1]{\textsuperscript{#2}}}
\newcommand{\fntext}[2]{\refstepcounter{fntext}\footnotetext[\thefntext]{\label{#1}#2}}
\newcommand{\refcite}[1]{Ref.\,\citenum{#1}\,}
\newtheoremstyle{theorem}{10pt}{6pt}{\itshape}{}{\bfseries}{.}{1em}
{\thmname{#1}\thmnumber{ #2}\thmnote{ (#3)}}
\theoremstyle{theorem}
\newtheorem{theorem}{Theorem}
\newtheorem{lemma}[theorem]{Lemma}

\newtheoremstyle{theoremname}{10pt}{6pt}{\itshape}{}{\bfseries}{.}{1em}
{\thmname{#3}}
\theoremstyle{theoremname}

\newtheoremstyle{dfn}{10pt}{6pt}{}{}{\bfseries}{.}{1em}{}
\theoremstyle{dfn}
\newtheorem{definition}[theorem]{Definition}

\newtheoremstyle{notation}{10pt}{6pt}{}{}{\bfseries}{.}{1em}{\thmname{#1}}
\theoremstyle{notation}

\newtheorem*{remark}{Remark}

\newtheoremstyle{convention}{10pt}{6pt}{}{}{\itshape}{:}{1em}{\thmname{#1}}
\theoremstyle{convention}

\newtheoremstyle{proof}{0mm}{6pt}{}{}{\bfseries\itshape}{ ---}{1em}
{\thmname{#1}\thmnumber{ #2}\thmnote{ #3}}
\theoremstyle{proof}
\newtheorem*{pf}{Proof}
\renewenvironment{proof}{\begin{pf}}{\qed\end{pf}}

\makeatletter
\newcommand*\NoIndentAfterEnv[1]{%
    \AfterEndEnvironment{#1}{\par\@afterindentfalse\@afterheading}}
\makeatother
\NoIndentAfterEnv{theorem}
\NoIndentAfterEnv{lemma}
\NoIndentAfterEnv{proposition}
\NoIndentAfterEnv{corollary}
\NoIndentAfterEnv{conjecture}
\NoIndentAfterEnv{theoremname}
\NoIndentAfterEnv{definition}
\NoIndentAfterEnv{notation}
\NoIndentAfterEnv{remark}
\NoIndentAfterEnv{convention}
\NoIndentAfterEnv{proof}

\def \st				{\,:\;}                                   
\def \textiff			{\;\text{iff}\;}                          
\def \union				{\,\cup\,}                              
\def \inter				{\,\cap\,}                              
\def \ge		{\geqslant}                                   
\def \le		{\leqslant}                                   
\renewcommand\preceq {\preccurlyeq}                       
\def \:= {\coloneqq}                                      

\def \R			{\mathbb R}                                   

\DeclareMathSymbol	{\inclnub}{\mathord}{letters}{44}

\DeclareMathSymbol	{\mapsnub}{\mathrel}{symbols}{55}
\def \to		{\longrightarrow}                             
\def \smallto		{\rightarrow}					                    
\def \mapsto		{\longmapsto}					                    
\def \smallmapsto	{\;{\mapsnub\rightarrow}\;}		              

\renewcommand \brace[1]	{\left\{ #1 \right\}}		       		
\newcommand  \set[2]	{\brace{#1 \;\left|\; #2 \right.}}  
\newcommand  \abs[1]	{\left| #1 \right|}		              

{\left[ \begin{smallmatrix} {}\\[.35mm]}
            {\\[.2mm]{}\\[-1.5mm]\end{smallmatrix} \right]}

{\left[\:\! \amsmatrix}
{\endamsmatrix \:\!\right]}

\newenvironment{piecewise}
{\left\{ \begin{array}{c@{\quad}l} }
        { \\\end{array} \right\}}
\renewenvironment{cases}                                  
{\begin{piecewise}}
        {\end{piecewise}}

\newenvironment{romanum}                                  
{\vspace{-1.5ex}
    \let\jrnbtemptopsep\topsep
    \setlength{\topsep}{-1ex}
    \begin{enumerate}
        \leftskip=1em
        \renewcommand\theenumi{\textup{(\textit{\roman{enumi}})}}
        \setlength{\itemsep}{-0.2ex}
        }{\end{enumerate}\vspace{-1.5ex}
    \setlength{\topsep}{\jrnbtemptopsep}}

\definecolor{zxg}{RGB}{216, 248, 216}
\newcommand\best[1]{\adjustbox{lap={0pt},raise={2pt}}{\colorbox{zxg}{\parbox[c][5pt][c]{18pt}{\centering\vfill\textcolor{black}{\textbf{#1}}\vfill}}}}

\newlength{\fw}
\newlength{\cw}
\newcolumntype{C}[1]{>{\centering\arraybackslash}m{#1}}
\newcommand\mT{\mathcal{T}}
\newcommand\mR{\mathcal{R}}

\newcommand{\XYm}{\ensuremath\normalfont\textrm{XY}\xspace}
\newcommand{\XZm}{\normalfont\normalfont\textrm{XZ}\xspace}
\newcommand{\YZm}{\normalfont\normalfont\textrm{YZ}\xspace}

\newcommand\symdiff{\;\Delta\;}

\def \condi {\textup{(\textit{\!\!\;i})}\xspace}
\def \condii {\textup{(\textit{\!\!\;i\!\!\;i})}\xspace}
\def \condiii {\textup{(\textit{\!\!\;i\!\!\;i\!\!\;i})}\xspace}

\def \IO {{\it I\text{-}O}\xspace}

\def \FullReduce {\texttt{full-reduce}\xspace}
\def \BasicOptimize {\texttt{basic-optimize}\xspace}
\def \FlowOpt {\texttt{flow-opt}\xspace}
\def \ZXHeur {\texttt{ZX-heur}\xspace}
\def \NRSCM {\texttt{NRSCM}\xspace}
\def \QFTOpt {\texttt{qft-opt}\xspace}
\newcommand \fs[2] {(\texttt{#1$^\mathbf{#2}$})\xspace}

\begin{document}

\title{Causal flow preserving optimisation of quantum circuits in the ZX-calculus}
\author{Calum Holker}
\affiliation{University of Oxford}
\maketitle

\begin{abstract}
    \noindent Optimising quantum circuits to minimise resource usage is crucial, especially with near-term hardware limited by quantum volume. This paper introduces an optimisation algorithm aiming to minimise non-Clifford gate count and two-qubit gate count by building on ZX-calculus-based strategies. By translating a circuit into a ZX-diagram it can be simplified before being extracted back into a circuit. We assert that simplifications preserve a graph-theoretic property called causal flow. This has the advantage that qubit lines are well defined throughout, permitting a trivial extraction procedure and in turn enabling the calculation of an individual transformation's impact on the resulting circuit. A general procedure for a decision strategy is introduced, inspired by an existing heuristic based method. Both phase teleportation and the neighbour unfusion rule are generalised. In particular, allowing unfusion of multiple neighbours is shown to lead to significant improvements in optimisation. When run on a set of benchmark circuits, the algorithm developed reduces the two-qubit gate count by an average of 19.8\%, beating both the previous best ZX-based strategy (14.6\%) and non-ZX strategy (18.5\%) at the time of publication. This lays a foundation for multiple avenues of improvement. A particularly effective strategy for optimising QFT circuits is also noted, resulting in exactly one two-qubit gate per non-Clifford gate.
\end{abstract}

\section{Introduction}
Quantum circuits represent computations as sequences of unitary evolutions of qubits called quantum gates. Universal fragments (restricted sets of operations from which any operator can be approximated to arbitrary accuracy) allow manageable hardware and efficient implementation. Clifford gates, an important subset of quantum operations, map the Pauli group onto itself. Examples include Pauli gates (X,Y,Z), Hadamard gates (H), phase gates (S) and CNOT gates \cite{nielsen2010quantum}. However, a universal set requires non-Clifford operations and at least one two-qubit gate \cite{divincenzo1995two}. The Clifford+T fragment, widely used in fault-tolerant computer design, consists of Clifford group operations and the T-gate, a $\pi/4$ rotation around the Z-axis of the Bloch sphere.

Quantum circuit optimisation translates a circuit into an equivalent one with fewer or simpler gates, thus reducing errors, decoherence effects and required resources for hardware implementation \cite{preskill2018quantum}. This paper focuses on two key metrics. First, two-qubit gates typically have a higher error rate than single qubit gates \cite{barends2014superconducting}. Second, non-Clifford gates have a much higher cost to implement on error-corrected devices, often over 150 times more expensive than other gates \cite{o2017quantum}. This paper refers to these metrics as \emph{2Q-count} and \emph{T-count} respectively, as many circuits use the minimal gate set of the Clifford+T fragment, but the methods are applicable on all gate sets.

Common optimisation approaches include gate substitutions, computation of small (pseudo-)normal forms \cite{kliuchnikov2013optimization,patel2008optimal}, optimisation of phase polynomials \cite{amy2014polynomial,heyfron2018efficient} or a combination \cite{abdessaied2014quantum,nam2018automated}. Some routines introduce ancilla qubits, but this paper focuses on ancilla-free methods. \refcite{duncan2020graph} introduced a new method using the ZX-calculus, involving a tensor-network style representation called a ZX-diagram. Circuits are translated into ZX-diagrams, simplified, then extracted back into circuits.
\begin{center}
    \begin{tikzpicture}
	\begin{pgfonlayer}{nodelayer}
		\node [style=none] (9) at (-13, 0) {\bf Quantum circuit};
		\node [style=none] (11) at (-9, 0.5) {};
		\node [style=none] (12) at (-3, 0.5) {};
		\node [style=none] (13) at (-6, 1.5) {\color{blue} \it translation};
		\node [style=none] (16) at (0.25, 0) {\bf ZX-diagram};
		\node [style=none] (21) at (-3, -0.5) {};
		\node [style=none] (22) at (-9, -0.5) {};
		\node [style=none] (23) at (-6, -1.5) {\color{blue} \it extraction};
		\node [style=none] (26) at (8, 0) {\color{red!40!blue!40!green} \it simplification};
		\node [style=none] (27) at (5, 0.5) {};
		\node [style=none] (28) at (5, -0.5) {};
		\node [style=none] (30) at (3.25, 0.5) {};
		\node [style=none] (31) at (3.25, -0.5) {};
	\end{pgfonlayer}
	\begin{pgfonlayer}{edgelayer}
		\draw [style=diredge, very thick, color=blue, in=165, out=15, looseness=0.75] (11.center) to (12.center);
		\draw [style=diredge, very thick, color=blue, in=345, out=-165, looseness=0.75] (21.center) to (22.center);
		\draw [style=none, very thick, color={red!40!blue!40!green}, in=135, out=30] (30.center) to (27.center);
		\draw [style=none, very thick, color={red!40!blue!40!green}, bend left=45] (27.center) to (28.center);
		\draw [style=diredge, very thick, color={red!40!blue!40!green}, in=-30, out=-135] (28.center) to (31.center);
	\end{pgfonlayer}
\end{tikzpicture}
\end{center}

The simplification procedure employs graph-theoretic transformations of local complementation and pivoting \cite{bouchet1988graphic,kotzig1968eulerian} to manipulate the diagram's underlying graph connectivity and properties. Many transformations break circuit structure but preserve a property called focused gflow \cite{browne2007generalized, mhalla2014graph}. Circuit extraction from general ZX-diagrams is a \#P-Hard problem \cite{de2022circuit}, but an efficient extraction procedure exists for diagrams with gflow \cite{backens2021there}. However, this approach often results in an extracted circuit with a \emph{higher} 2Q-count than the original for complex circuits. \refcite{staudacher2022reducing} introduced a heuristic-based decision strategy for applying transformations to reduce 2Q-count. However, due to the extraction procedure, there is inherent unpredictability between a transformation and its effect on the resulting circuit; this limits the accuracy of potential algorithms.

Circuit to circuit translations in the ZX-calculus were initially introduced in \refcite{fagan2019optimising} for basic transformations preserving a stricter condition called causal flow (\emph{cflow}) \cite{danos2006determinism}. ZX-diagrams with cflow have a direct circuit analogue and a trivial extraction procedure. This paper explores a simplification procedure preserving cflow, addressing extraction issues and increasing the mutual information between transformations and output circuits. The 2Q-count, gate count and qubit connectivity remain visible throughout. The presented algorithm, incorporating a ZX-calculus technique for reducing T-count \cite{kissinger2019reducing}, acts as a foundation for general decision strategies and future improvements. The particular effectiveness on 2Q-count reduction of a generalisation of the neighbour unfusion rule is also demonstrated. The result is an efficient optimisation algorithm which outperforms state-of-the-art optimisation routines such as those of \refcite{nam2018automated} on standard benchmark circuits. The algorithm is implemented in a fork\footnote{\label{pyzx}\url{https://github.com/calumholker/pyzx}} of the Python library PyZX \cite{kissinger2019pyzx}, with demo notebooks\footnote{\label{demo}\url{https://github.com/calumholker/pyzx/tree/master/demos/circuit_optimisation}} available.

This paper is organised as follows: \cref{sec:zx-calculus} introduces the ZX-calculus; \cref{sec:graph-like-ZX} provides the relevant precursors for applying graph-theory on ZX-diagrams; \cref{sec:simplify-zx} presents core simplification methods; \cref{sec:extraction-and-flow} introduces flow and circuit extraction; \cref{sec:flow-preservation} dicusses cflow preserving transformations; \cref{sec:optimisation} describes the optimisation algorithm, with results in \cref{subsec:results}; and \cref{sec:conclusion} concludes and discusses future work.
\section{The ZX-calculus}\label{sec:zx-calculus}
This section gives a brief introduction to the ZX-calculus. For a more detailed overview see \refcite{PQP}.

The ZX-calculus is a tensor-network style graphical language for representing general linear maps between Hilbert spaces as ZX-diagrams. These diagrams consist of spiders and wires, analogous to gates and wires in quantum circuits. Unlike gates, which require unitary maps with equal numbers of inputs and outputs, spiders are more flexible linear maps. The green Z-spider and red X-spider are defined in terms of the Z-basis $\ket{0},\ket{1} $ and X-basis $\ket{\pm} = \frac{1}{\sqrt{2}} (\ket{0} \pm \ket{1})$, respectively:
\begin{equation}
    \begin{tikzpicture}
	\begin{pgfonlayer}{nodelayer}
		\node [style=Z phase dot] (0) at (0, 0) {$\alpha$};
		\node [style=none] (1) at (1.25, 1) {};
		\node [style=none] (2) at (-1.25, 1) {};
		\node [style=none] (3) at (-1.25, -1) {};
		\node [style=none] (4) at (1.25, -1) {};
		\node [style=none] (5) at (1.25, 0.5) {};
		\node [style=none] (6) at (-1.25, 0.5) {};
		\node [style=none, rotate=90] (7) at (-1, -0.25) {...};
		\node [style=none, rotate=90] (8) at (1, -0.25) {...};
	\end{pgfonlayer}
	\begin{pgfonlayer}{edgelayer}
		\draw [in=-141, out=0, looseness=0.75] (3.center) to (0);
		\draw [in=180, out=-39, looseness=0.75] (0) to (4.center);
		\draw [in=180, out=22, looseness=0.75] (0) to (5.center);
		\draw [in=180, out=39, looseness=0.75] (0) to (1.center);
		\draw [in=0, out=158, looseness=0.75] (0) to (6.center);
		\draw [in=141, out=0, looseness=0.75] (2.center) to (0);
	\end{pgfonlayer}
\end{tikzpicture} \ \:= \ \ketbra{\textrm{$0$...$0$}}{\textrm{$0$...$0$}} +
    e^{i \alpha} \ketbra{\textrm{$1$...$1$}}{\textrm{$1$...$1$}} \hfill
    \qquad
    \hfill \begin{tikzpicture}
	\begin{pgfonlayer}{nodelayer}
		\node [style=X phase dot] (0) at (0, 0) {$\alpha$};
		\node [style=none] (1) at (1.25, 1) {};
		\node [style=none] (2) at (-1.25, 1) {};
		\node [style=none] (3) at (-1.25, -1) {};
		\node [style=none] (4) at (1.25, -1) {};
		\node [style=none] (5) at (1.25, 0.5) {};
		\node [style=none] (6) at (-1.25, 0.5) {};
		\node [style=none, rotate=90] (7) at (-1, -0.25) {...};
		\node [style=none, rotate=90] (8) at (1, -0.25) {...};
	\end{pgfonlayer}
	\begin{pgfonlayer}{edgelayer}
		\draw [in=-141, out=0, looseness=0.75] (3.center) to (0);
		\draw [in=180, out=-39, looseness=0.75] (0) to (4.center);
		\draw [in=180, out=22, looseness=0.75] (0) to (5.center);
		\draw [in=180, out=39, looseness=0.75] (0) to (1.center);
		\draw [in=0, out=158, looseness=0.75] (0) to (6.center);
		\draw [in=141, out=0, looseness=0.75] (2.center) to (0);
	\end{pgfonlayer}
\end{tikzpicture} \ \:= \ \ketbra{\textrm{$+$...$+$}}{\textrm{$+$...$+$}} +
    e^{i \alpha} \ketbra{\textrm{$-$...$-$}}{\textrm{$-$...$-$}}
\end{equation}
Spiders can be composed by connecting the input (left-entering) wires of one to the output (right-exiting) wires of another, and tensor products are formed by stacking them. A ZX-diagram with \emph{n} inputs and \emph{m} outputs, which is any combination of these compositions and tensor products, maps $(\mathbb C^2)^{\otimes n} \smallto (\mathbb C^2)^{\otimes m}$. Diagrams without inputs represent states in $(\mathbb C^2)^{\otimes m}$.

It is convenient to introduce an additional generator for the Hadamard gate, defined in terms of its decomposition into spiders and represented by a yellow box. The representation of a Hadamard between two spiders is often simplified as a blue dashed edge, called a \emph{Hadamard edge}.
\begin{align}
    \begin{tikzpicture}
	\begin{pgfonlayer}{nodelayer}
		\node [style=none] (0) at (-2.75, 0) {};
		\node [style=none] (1) at (-0.75, 0) {};
		\node [style=hadamard] (2) at (-1.75, 0) {};
		\node [style=none] (3) at (0, 0) {$=$};
		\node [style=none] (4) at (0.75, 0) {};
		\node [style=Z phase dot] (5) at (1.5, 0) {$\frac\pi2$};
		\node [style=X phase dot] (6) at (2.5, 0) {$\frac\pi2$};
		\node [style=Z phase dot] (7) at (3.5, 0) {$\frac\pi2$};
		\node [style=none] (8) at (4.25, 0) {};
	\end{pgfonlayer}
	\begin{pgfonlayer}{edgelayer}
		\draw (0.center) to (2);
		\draw (2) to (1.center);
		\draw (4.center) to (8.center);
	\end{pgfonlayer}
\end{tikzpicture} &  & \begin{tikzpicture}
	\begin{pgfonlayer}{nodelayer}
		\node [style=none] (0) at (-4.5, -0.75) {};
		\node [style=Z dot] (4) at (-3.75, 0) {};
		\node [style=none] (6) at (-4.5, 0.75) {};
		\node [style=none] (8) at (0, 0) {$:=$};
		\node [style=none] (26) at (-1, -0.75) {};
		\node [style=Z dot] (27) at (-1.75, 0) {};
		\node [style=none] (28) at (-1, 0.75) {};
		\node [style=none] (29) at (1, -0.75) {};
		\node [style=Z dot] (30) at (1.75, 0) {};
		\node [style=none] (31) at (1, 0.75) {};
		\node [style=none] (32) at (4.5, -0.75) {};
		\node [style=Z dot] (33) at (3.75, 0) {};
		\node [style=none] (34) at (4.5, 0.75) {};
		\node [style=hadamard] (35) at (2.75, 0) {};
		\node [style=none, rotate=90] (36) at (-4.5, 0) {...};
		\node [style=none, rotate=90] (37) at (-1, 0) {...};
		\node [style=none, rotate=90] (38) at (1, 0) {...};
		\node [style=none, rotate=90] (39) at (4.5, 0) {...};
	\end{pgfonlayer}
	\begin{pgfonlayer}{edgelayer}
		\draw (0.center) to (4);
		\draw (6.center) to (4);
		\draw (26.center) to (27);
		\draw (28.center) to (27);
		\draw [style=hadamard edge] (4) to (27);
		\draw (29.center) to (30);
		\draw (31.center) to (30);
		\draw (32.center) to (33);
		\draw (34.center) to (33);
		\draw (30) to (33);
	\end{pgfonlayer}
\end{tikzpicture}
\end{align}
Quantum circuits can be easily translated into ZX-diagrams. A circuit represented in the universal gate set $\brace{\textit{CNOT},Z_\alpha,H}$ has direct spider analogues:
\begin{align}
    \textit{CNOT} \;=\; \begin{tikzpicture}
	\begin{pgfonlayer}{nodelayer}
		\node [style=Z dot] (0) at (0, 0.5) {};
		\node [style=none] (1) at (1.25, 0.5) {};
		\node [style=X dot] (2) at (0, -0.5) {};
		\node [style=none] (3) at (-1.25, 0.5) {};
		\node [style=none] (4) at (1.25, -0.5) {};
		\node [style=none] (5) at (-1.25, -0.5) {};
	\end{pgfonlayer}
	\begin{pgfonlayer}{edgelayer}
		\draw (0) to (2);
		\draw (0) to (3.center);
		\draw (0) to (1.center);
		\draw (5.center) to (2);
		\draw (2) to (4.center);
	\end{pgfonlayer}
\end{tikzpicture}      &  &
    Z_\alpha \;=\; \begin{tikzpicture}
	\begin{pgfonlayer}{nodelayer}
		\node [style=Z phase dot] (0) at (0, 0) {$\alpha$};
		\node [style=none] (1) at (-1.25, 0) {};
		\node [style=none] (2) at (1.25, 0) {};
	\end{pgfonlayer}
	\begin{pgfonlayer}{edgelayer}
		\draw (1.center) to (0);
		\draw (0) to (2.center);
	\end{pgfonlayer}
\end{tikzpicture} &  &
    H \;=\; \begin{tikzpicture}
	\begin{pgfonlayer}{nodelayer}
		\node [style=none] (0) at (-0.25, 0) {};
		\node [style=none] (1) at (1.75, 0) {};
		\node [style=hadamard] (2) at (0.75, 0) {};
	\end{pgfonlayer}
	\begin{pgfonlayer}{edgelayer}
		\draw (0.center) to (2);
		\draw (2) to (1.center);
	\end{pgfonlayer}
\end{tikzpicture}
\end{align}
A key property of ZX-diagrams is that \textbf{only connectivity matters}, meaning diagrams can be deformed in any way by moving spiders around the plane or bending wires, and the linear map will remain equal provided the order of inputs and outputs is preserved. The \emph{rules} of the ZX-calculus provide a core set of rewrite equations under which the linear map is preserved (\cref{fig:zx-rules}). In this rule set, non-zero scalar factors are neglected; this is suitable for this paper as we deal only with the unitary maps of circuits, in which scalar factors manifest as inconsequential global phases.

\vspace{15pt}
\begin{figure}[ht]
    \centering
    \begin{tikzpicture}
	\begin{pgfonlayer}{nodelayer}
		\node [style=Z phase dot] (0) at (-1.5, -0.625) {$\beta$};
		\node [style=none] (1) at (-0.5, -0.375) {};
		\node [style=none] (2) at (-3.25, -0.375) {};
		\node [style=none] (3) at (-3.25, -1.625) {};
		\node [style=none] (4) at (-0.5, -1.625) {};
		\node [style=none, rotate=90] (7) at (-3.5, -1) {...};
		\node [style=none, rotate=90] (8) at (-0.5, -1) {...};
		\node [style=Z phase dot] (10) at (-3, 0.875) {$\alpha$};
		\node [style=none] (11) at (-1, 1.625) {};
		\node [style=none] (12) at (-4, 1.625) {};
		\node [style=none] (13) at (-1, 0.375) {};
		\node [style=none, rotate=90] (14) at (-1, 1) {...};
		\node [style=none] (16) at (-4, 0.375) {};
		\node [style=none, rotate=90] (17) at (-3.75, 1) {...};
		\node [style=none] (18) at (-0.5, 0.375) {};
		\node [style=none] (19) at (-0.5, 1.625) {};
		\node [style=none] (22) at (-4, -1.625) {};
		\node [style=none] (23) at (-4, -0.375) {};
		\node [style=none] (24) at (0.5, 0.25) {$=$};
		\node [style=none, rotate=90] (25) at (-2.35, 0.05) {...};
		\node [style=none] (26) at (1.5, 1.25) {};
		\node [style=none, rotate=90] (28) at (2, 0) {...};
		\node [style=none] (29) at (5, -1.25) {};
		\node [style=none, rotate=90] (30) at (4.75, 0) {...};
		\node [style=Z phase dot] (32) at (3.25, 0) {$\ \alpha\!+\!\beta\ $};
		\node [style=none] (33) at (5, 1.25) {};
		\node [style=none] (34) at (1.5, -1.25) {};
		\node [style=none] (35) at (0.5, 1) {\spiderrule};
		\node [style=Z phase dot] (36) at (2.75, -4.25) {$\ -\alpha\ $};
		\node [style=none] (37) at (5, -3.125) {};
		\node [style=none] (38) at (0.5, -4.25) {$=$};
		\node [style=none] (39) at (1.5, -4.25) {};
		\node [style=none] (40) at (5, -5.375) {};
		\node [style=X phase dot] (42) at (4.25, -5.375) {$\pi$};
		\node [style=X phase dot] (43) at (-2.75, -4.25) {$\pi$};
		\node [style=none] (44) at (-0.25, -3.125) {};
		\node [style=Z phase dot] (45) at (-1.75, -4.25) {$\alpha$};
		\node [style=none, rotate=90] (46) at (-0.5, -4.25) {...};
		\node [style=none, rotate=90] (48) at (4.25, -4.25) {...};
		\node [style=none] (50) at (-3.5, -4.25) {};
		\node [style=none] (51) at (-0.25, -5.375) {};
		\node [style=X phase dot] (52) at (4.25, -3.125) {$\pi$};
		\node [style=none] (53) at (0.5, -3.5) {\pirule};
		\node [style=X dot] (54) at (12, -3.375) {};
		\node [style=none, rotate=90] (55) at (12.75, -4.25) {...};
		\node [style=Z phase dot] (57) at (8.25, -4.25) {$\alpha$};
		\node [style=none] (58) at (9.5, -5.125) {};
		\node [style=none] (59) at (10.5, -4.25) {$=$};
		\node [style=none] (60) at (13, -5.125) {};
		\node [style=none] (62) at (13, -3.375) {};
		\node [style=none, rotate=90] (63) at (9.25, -4.25) {...};
		\node [style=X dot] (65) at (12, -5.125) {};
		\node [style=none] (66) at (10.5, -3.5) {\copyrule};
		\node [style=none] (67) at (9.5, -3.375) {};
		\node [style=X dot] (68) at (7.25, -4.25) {};
		\node [style=hadamard] (69) at (7.5, 1) {};
		\node [style=none, rotate=90] (72) at (7.5, 0) {...};
		\node [style=none] (73) at (10.75, 0) {$=$};
		\node [style=hadamard] (74) at (7.5, -1) {};
		\node [style=none] (77) at (7, 1) {};
		\node [style=none] (81) at (11.75, -1) {};
		\node [style=none] (82) at (7, -1) {};
		\node [style=none, rotate=90] (83) at (12, 0) {...};
		\node [style=none] (84) at (11.75, 1) {};
		\node [style=none] (86) at (10.75, 0.75) {\hadamardrule};
		\node [style=Z dot] (87) at (16.5, 0.75) {};
		\node [style=none] (88) at (15.75, 0.75) {};
		\node [style=none] (89) at (17.25, 0.75) {};
		\node [style=none] (90) at (19.25, 0.75) {};
		\node [style=none] (91) at (20.25, 0.75) {};
		\node [style=none] (92) at (18.25, 1.5) {\identityrule};
		\node [style=none] (93) at (18.25, 0.75) {$=$};
		\node [style=none] (94) at (18.25, -1.75) {$=$};
		\node [style=none] (95) at (18.25, -1) {\hhrule};
		\node [style=none] (96) at (20.25, -1.75) {};
		\node [style=none] (97) at (17.25, -1.75) {};
		\node [style=none] (98) at (15.75, -1.75) {};
		\node [style=none] (99) at (19.25, -1.75) {};
		\node [style=hadamard] (100) at (16.25, -1.75) {};
		\node [style=hadamard] (101) at (16.75, -1.75) {};
		\node [style=none] (102) at (18.25, -3.5) {\bialgrule};
		\node [style=X dot] (103) at (21.25, -5) {};
		\node [style=none] (104) at (22, -3.5) {};
		\node [style=X dot] (105) at (15.5, -4.25) {};
		\node [style=Z dot] (106) at (16.75, -4.25) {};
		\node [style=none] (107) at (19, -5) {};
		\node [style=none] (108) at (14.75, -3.5) {};
		\node [style=none] (109) at (19, -3.5) {};
		\node [style=Z dot] (110) at (19.75, -3.5) {};
		\node [style=none] (111) at (17.5, -3.5) {};
		\node [style=none] (112) at (18.25, -4.25) {$=$};
		\node [style=none] (113) at (14.75, -5) {};
		\node [style=X dot] (114) at (21.25, -3.5) {};
		\node [style=none] (115) at (22, -5) {};
		\node [style=Z dot] (116) at (19.75, -5) {};
		\node [style=none] (117) at (17.5, -5) {};
		\node [style=hadamard] (118) at (9.5, 1) {};
		\node [style=none, rotate=90] (120) at (9.5, 0) {...};
		\node [style=hadamard] (121) at (9.5, -1) {};
		\node [style=Z phase dot] (122) at (8.5, 0) {$\alpha$};
		\node [style=none] (123) at (10, 1) {};
		\node [style=none] (126) at (10, -1) {};
		\node [style=X phase dot] (127) at (13, 0) {$\alpha$};
		\node [style=none] (128) at (14.25, -1) {};
		\node [style=none, rotate=90] (129) at (14, 0) {...};
		\node [style=none] (130) at (14.25, 1) {};
	\end{pgfonlayer}
	\begin{pgfonlayer}{edgelayer}
		\draw [in=-141, out=0, looseness=0.75] (3.center) to (0);
		\draw [in=180, out=-39, looseness=0.75] (0) to (4.center);
		\draw [in=180, out=39, looseness=0.75] (0) to (1.center);
		\draw [in=180, out=0] (2.center) to (0);
		\draw [in=-141, out=0, looseness=0.75] (16.center) to (10);
		\draw [in=180, out=0] (10) to (13.center);
		\draw [in=180, out=39, looseness=0.75] (10) to (11.center);
		\draw [in=141, out=0, looseness=0.75] (12.center) to (10);
		\draw [bend left=45] (10) to (0);
		\draw (11.center) to (19.center);
		\draw (13.center) to (18.center);
		\draw (23.center) to (2.center);
		\draw (22.center) to (3.center);
		\draw [bend left=45] (0) to (10);
		\draw [in=-120, out=0] (34.center) to (32);
		\draw [in=180, out=-60] (32) to (29.center);
		\draw [in=180, out=60] (32) to (33.center);
		\draw [in=120, out=0] (26.center) to (32);
		\draw [in=180, out=-75, looseness=0.75] (45) to (51.center);
		\draw [in=180, out=75, looseness=0.75] (45) to (44.center);
		\draw [in=180, out=0, looseness=0.50] (43) to (45);
		\draw (50.center) to (43);
		\draw [in=-60, out=180, looseness=0.75] (42) to (36);
		\draw [in=0, out=180, looseness=0.75] (36) to (39.center);
		\draw [in=60, out=180, looseness=0.75] (52) to (36);
		\draw (37.center) to (52);
		\draw (40.center) to (42);
		\draw [in=180, out=-75, looseness=0.75] (57) to (58.center);
		\draw [in=180, out=75, looseness=0.75] (57) to (67.center);
		\draw [in=180, out=0, looseness=0.50] (68) to (57);
		\draw (62.center) to (54);
		\draw (60.center) to (65);
		\draw (77.center) to (69);
		\draw (82.center) to (74);
		\draw (88.center) to (89.center);
		\draw (90.center) to (91.center);
		\draw (98.center) to (97.center);
		\draw (99.center) to (96.center);
		\draw (116) to (114);
		\draw (103) to (110);
		\draw (110) to (114);
		\draw (116) to (103);
		\draw (115.center) to (103);
		\draw (107.center) to (116);
		\draw (110) to (109.center);
		\draw (114) to (104.center);
		\draw [bend right] (113.center) to (105);
		\draw [bend right] (105) to (108.center);
		\draw (105) to (106);
		\draw [bend right] (106) to (117.center);
		\draw [bend left] (106) to (111.center);
		\draw [in=-60, out=180, looseness=0.75] (121) to (122);
		\draw [in=75, out=180, looseness=0.75] (118) to (122);
		\draw (123.center) to (118);
		\draw (126.center) to (121);
		\draw [in=-120, out=0, looseness=0.75] (74) to (122);
		\draw [in=105, out=0, looseness=0.75] (69) to (122);
		\draw [in=180, out=-60, looseness=0.75] (127) to (128.center);
		\draw [in=180, out=75, looseness=0.75] (127) to (130.center);
		\draw [in=0, out=-120, looseness=0.75] (127) to (81.center);
		\draw [in=0, out=105, looseness=0.75] (127) to (84.center);
		\draw [in=105, out=-15, looseness=0.75] (10) to (0);
	\end{pgfonlayer}
\end{tikzpicture}
    \caption{A presentation of the core rules for the ZX-calculus as rewrite equations \cite{duncan2020graph}. These hold $\forall \alpha, \beta \in [0, 2 \pi)$, and due to \HadamardRule and \HHRule all rules also hold with the colours interchanged. Note `...' should be read as `0 or more'.} \label{fig:zx-rules}
    \vspace{0pt}
\end{figure}
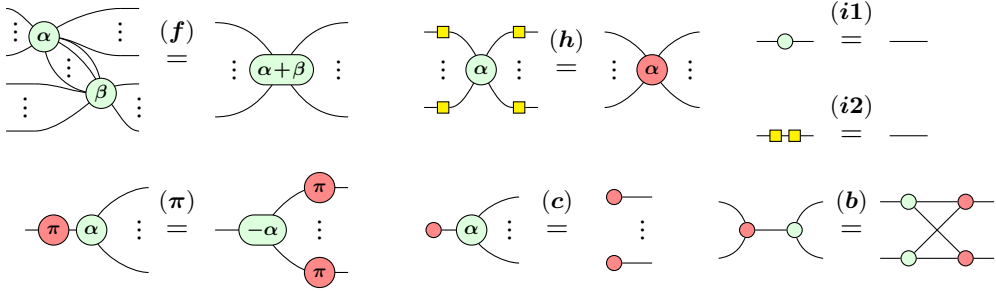

\begin{remark}
    The ZX-calculus is \emph{universal}, meaning any linear map can be represented as a ZX-diagram. The core rules presented in \cref{fig:zx-rules} are \emph{complete} for the Clifford fragment \cite{backens2014zx}, meaning any two ZX-diagrams representing the same underlying linear map can be transformed into each other through a series of rewrites. This pertains to diagrams comprising spiders with phases which are multiples of $\pi/2$ (\emph{Clifford spiders}). Extensions of the calculus have proven to be complete for the Clifford+T fragment \cite{jeandel2018complete} and for all ZX-diagrams \cite{hadzihasanovic2018two, vilmart2018near}.
\end{remark}\vspace{10pt}
\section{Graph-like ZX-diagrams}\label{sec:graph-like-ZX}
\begin{definition}
    A \emph{graph} $G$ is a tuple $(V,E)$ which represent sets of vertices and edges, respectively. Two vertices $u, v$ are adjacent ($u \sim v$) if connected by an edge $(u,v)$. $N_G(u)$ denotes the neighbours of $u$. We use the convention that graphs cannot contain loops on a vertex or multiple edges between two vertices.
\end{definition}

\begin{definition}
    An \emph{open graph} $(G,I,O)$ consists of an undirected graph and subsets $I,O \subseteq V$ denoting \emph{inputs} and \emph{outputs}. Internal vertices are in $\bar{I} \inter \bar{O}$ and boundary vertices are in $I \union O$.
\end{definition}

After transforming a circuit into a ZX-diagram, it can be transformed into a graph-like ZX-diagram.

\begin{definition}\label{def:graph-like}
    A given ZX-diagram is \emph{graph-like} if it satisfies the following conditions:
    \begin{romanum}
        \item All spiders are Z-spiders
        \item All edges between spiders are Hadamard edges
        \item There are no parallel Hadamard edges or self-loops
        \item Every Z-spider is connected to, at most, one input or one output
        \item Every input or output is connected to a Z-spider
    \end{romanum}
\end{definition}

Graph-like ZX-diagrams extend the concept of graph-states \cite{hein2006entanglement}, defined as a graph-like diagram where every spider connects to an output with no non-zero phases.
\begin{lemma}[\citenum{duncan2020graph}]
    \label{lem:graph-like-trans}
    Every ZX-diagram is equal to a graph-like ZX-diagram.
\end{lemma}
\begin{proof}
    Each step in the following procedure satisfies the respective condition in \cref{def:graph-like}.
    \begin{romanum}
        \item Apply \HadamardRule to turn all X-spiders into Z-spiders.
        \item Apply \HHRule to remove pairs of Hadamards, and \SpiderRule to fuse spiders not connected via a Hadamard.
        \item Remove all parallel Hadamard edges and self loops using the following three rules \cite{duncan2020graph}.
        \begin{align*}
            \begin{tikzpicture}
	\begin{pgfonlayer}{nodelayer}
		\node [style=Z phase dot] (0) at (-1.25, 0) {$\alpha$};
		\node [style=Z phase dot] (1) at (1.25, 0) {$\beta$};
		\node [style=none] (2) at (-2.25, 0.75) {};
		\node [style=none] (3) at (-2.25, -0.75) {};
		\node [style=none] (4) at (2.25, 0.75) {};
		\node [style=none] (5) at (2.25, -0.75) {};
		\node [style=none, rotate=90] (6) at (-2.25, 0) {...};
		\node [style=none, rotate=90] (7) at (2.25, 0) {...};
		\node [style=hadamard] (8) at (0, 0.75) {};
		\node [style=hadamard] (9) at (0, -0.75) {};
		\node [style=none] (10) at (3.25, 0) {=};
		\node [style=Z phase dot] (11) at (5.5, 0) {$\alpha$};
		\node [style=Z phase dot] (12) at (7, 0) {$\beta$};
		\node [style=none] (13) at (4.5, 0.75) {};
		\node [style=none] (14) at (4.5, -0.75) {};
		\node [style=none] (15) at (8, 0.75) {};
		\node [style=none] (16) at (8, -0.75) {};
		\node [style=none, rotate=90] (17) at (4.5, 0) {...};
		\node [style=none, rotate=90] (18) at (8, 0) {...};
	\end{pgfonlayer}
	\begin{pgfonlayer}{edgelayer}
		\draw (0) to (2.center);
		\draw (0) to (3.center);
		\draw (1) to (4.center);
		\draw (1) to (5.center);
		\draw [in=-180, out=60] (0) to (8);
		\draw [in=120, out=0] (8) to (1);
		\draw [in=0, out=-120] (1) to (9);
		\draw [in=-60, out=180] (9) to (0);
		\draw (11) to (13.center);
		\draw (11) to (14.center);
		\draw (12) to (15.center);
		\draw (12) to (16.center);
	\end{pgfonlayer}
\end{tikzpicture}  &  &
            \begin{tikzpicture}
	\begin{pgfonlayer}{nodelayer}
		\node [style=Z phase dot] (0) at (-1.75, 0) {$\alpha$};
		\node [style=none] (1) at (-2.5, -1) {};
		\node [style=none] (2) at (-1, -1) {};
		\node [style=none] (3) at (-1.75, 1) {};
		\node [style=none] (4) at (-1.75, -1) {...};
		\node [style=none] (5) at (0, 0) {=};
		\node [style=none] (6) at (1.5, -1) {...};
		\node [style=none] (7) at (0.75, -1) {};
		\node [style=none] (8) at (2.25, -1) {};
		\node [style=Z phase dot] (9) at (1.5, 0) {$\alpha$};
	\end{pgfonlayer}
	\begin{pgfonlayer}{edgelayer}
		\draw (1.center) to (0);
		\draw (0) to (2.center);
		\draw [in=360, out=15, looseness=1.50] (0) to (3.center);
		\draw [in=540, out=165, looseness=1.50] (0) to (3.center);
		\draw (7.center) to (9);
		\draw (9) to (8.center);
	\end{pgfonlayer}
\end{tikzpicture} &  &
            \begin{tikzpicture}
	\begin{pgfonlayer}{nodelayer}
		\node [style=Z phase dot] (0) at (-1.75, 0) {$\alpha$};
		\node [style=none] (1) at (-2.5, -1) {};
		\node [style=none] (2) at (-1, -1) {};
		\node [style=none] (3) at (-1.75, 1) {};
		\node [style=none] (4) at (-1.75, -1) {...};
		\node [style=none] (5) at (0, 0) {=};
		\node [style=none] (6) at (1.75, -1) {...};
		\node [style=none] (7) at (1, -1) {};
		\node [style=none] (8) at (2.5, -1) {};
		\node [style=Z phase dot] (9) at (1.75, 0) {$\ \alpha + \pi\ $};
		\node [style=hadamard] (44) at (-1.75, 1) {};
	\end{pgfonlayer}
	\begin{pgfonlayer}{edgelayer}
		\draw (1.center) to (0);
		\draw (0) to (2.center);
		\draw [in=360, out=15, looseness=1.50] (0) to (3.center);
		\draw [in=540, out=165, looseness=1.50] (0) to (3.center);
		\draw (7.center) to (9);
		\draw (9) to (8.center);
	\end{pgfonlayer}
\end{tikzpicture}
        \end{align*}
        \item For Z-spiders connected to multiple inputs/outputs, apply \IdentityRule and \HHRule in reverse.
        \ctikzfig{ident-graph-form}
        \item For inputs/outputs connected by a bare wire or sole Hadamard gate, apply \IdentityRule and \HHRule in reverse:
        \ctikzfig{ident-graph-form-2}
    \end{romanum}\vspace{-10pt}
\end{proof}
Translating a ZX-diagram into graph-like form enables the description of its structure by its underlying open graph, allowing us to evaluate graph-theoretic properties of the diagram. This is crucial for the application of graph-theoretic techniques, utilised in \cref{sec:extraction-and-flow}.

\begin{definition}
    The \emph{underlying open graph} $G(D)$ of a ZX-diagram $D$ consists of spiders as vertices, Hadamard edges as edges, and subsets of spiders connected to inputs and outputs as $I$ and $O$, respectively.
\end{definition}\vspace{10pt}
\section{Simplifying ZX-diagrams}\label{sec:simplify-zx}
This section outlines core simplifications rules for graph-like ZX-diagrams, maintaining their graph-like nature and derived using rules from \cref{fig:zx-rules}. The first is \emph{identity-fusion} which eliminates identity spiders by sequentially applying \IdentityRule, \HHRule and \SpiderRule. Equivalent series of transformations are applied throughout the procedure in \refcite{duncan2020graph}, but we define it explicitly as a rule which preserves the diagram's graph-like form.

\begin{equation}
    \begin{tikzpicture}
	\begin{pgfonlayer}{nodelayer}
		\node [style=Z phase dot] (2) at (-7, 0) {$\alpha$};
		\node [style=Z phase dot] (4) at (-3, 0) {$\beta$};
		\node [style=none] (5) at (-8, 0.75) {};
		\node [style=none] (6) at (-8, -0.75) {};
		\node [style=none, rotate=90] (10) at (-7.95, 0) {...};
		\node [style=none] (16) at (0, 0) {$=$};
		\node [style=none] (48) at (-5, 0.75) {\color{red} $*$};
		\node [style=Z dot] (49) at (-5, 0) {};
		\node [style=none] (50) at (-2, 0.75) {};
		\node [style=none] (51) at (-2, -0.75) {};
		\node [style=none, rotate=90] (52) at (-2.05, 0) {...};
		\node [style=Z phase dot] (53) at (3.5, 0) {$\ \alpha + \beta\ $};
		\node [style=none] (54) at (5, 0.75) {};
		\node [style=none] (55) at (5, -0.75) {};
		\node [style=none] (56) at (2, -0.75) {};
		\node [style=none] (57) at (2, 0.75) {};
		\node [style=none, rotate=90] (58) at (2, 0) {...};
		\node [style=none, rotate=90] (59) at (5, 0) {...};
		\node [style=none] (60) at (0, 0.75) {\idfuserule};
		\node [style=none] (61) at (3, 0.25) {};
		\node [style=none] (62) at (3, -0.25) {};
		\node [style=none] (63) at (4, 0.25) {};
		\node [style=none] (64) at (4, -0.25) {};
	\end{pgfonlayer}
	\begin{pgfonlayer}{edgelayer}
		\draw (5.center) to (2);
		\draw (6.center) to (2);
		\draw [style=hadamard edge] (49) to (4);
		\draw [style=hadamard edge] (2) to (49);
		\draw (4) to (50.center);
		\draw (4) to (51.center);
		\draw (63.center) to (54.center);
		\draw (64.center) to (55.center);
		\draw (56.center) to (62.center);
		\draw (57.center) to (61.center);
	\end{pgfonlayer}
\end{tikzpicture}
\end{equation}\vspace{0pt}

\noindent The graph-theoretic transformations of \emph{local complementation} and \emph{pivoting} \cite{bouchet1988graphic,kotzig1968eulerian} are then utilised, modified for application on graph-like ZX-diagrams \cite{duncan2020graph}.
\begin{definition}
    The \emph{local complementation} of a graph $G$ on $u \in V(G)$,
    is the graph $G \star u$ with $V(G \star u) \:= V(G)$ and $\forall v,w \in N_G(u)$ we have $(v,w) \in E(G \star u) \textiff (v,w) \notin E(G)$. All other edges are unchanged.
\end{definition}
This corresponds to complementing edges between all neighbours of a vertex, illustrated below \cite{duncan2020graph}.
\begin{equation*}
    G \quad \begin{tikzpicture}
	\begin{pgfonlayer}{nodelayer}
		\node [style=vertex] (0) at (-1, 1) {};
		\node [style=vertex] (1) at (-1, -1) {};
		\node [style=vertex] (2) at (1, -1) {};
		\node [style=vertex] (3) at (1, 1) {};
		\node [style=none] (4) at (-1.5, 1) {$a$};
		\node [style=none] (5) at (1.5, 1) {$b$};
		\node [style=none] (6) at (1.5, -1) {$d$};
		\node [style=none] (7) at (-1.5, -1) {$c$};
	\end{pgfonlayer}
	\begin{pgfonlayer}{edgelayer}
		\draw (0) to (3);
		\draw (0) to (2);
		\draw (3) to (2);
		\draw (0) to (1);
	\end{pgfonlayer}
\end{tikzpicture} \qquad\qquad G \star a \quad \begin{tikzpicture}
	\begin{pgfonlayer}{nodelayer}
		\node [style=vertex] (0) at (-1, 1) {};
		\node [style=vertex] (1) at (-1, -1) {};
		\node [style=vertex] (2) at (1, -1) {};
		\node [style=vertex] (3) at (1, 1) {};
		\node [style=none] (4) at (-1.5, 1) {$a$};
		\node [style=none] (5) at (1.5, 1) {$b$};
		\node [style=none] (6) at (1.5, -1) {$d$};
		\node [style=none] (7) at (-1.5, -1) {$c$};
	\end{pgfonlayer}
	\begin{pgfonlayer}{edgelayer}
		\draw (0) to (3);
		\draw (0) to (2);
		\draw (0) to (1);
		\draw (3) to (1);
		\draw (1) to (2);
	\end{pgfonlayer}
\end{tikzpicture}
\end{equation*}
\noindent In the following rule \cite{duncan2020graph} parallel edges cancel, thus it is equivalent to performing a local complementation of the underlying open graph on the marked vertex, removing it, and updating the phases as shown.

\vspace{-10pt}\begin{equation}
    \begin{tikzpicture}
	\begin{pgfonlayer}{nodelayer}
		\node [style=Z phase dot] (0) at (-5.5, 1) {$\ \pm\frac\pi2\ $};
		\node [style=Z phase dot] (2) at (-8, 0.5) {$\alpha_1$};
		\node [style=Z phase dot] (4) at (-2.75, 0.5) {$\alpha_n$};
		\node [style=none] (5) at (-9, -0.5) {};
		\node [style=none] (6) at (-8, -0.5) {};
		\node [style=none] (7) at (-2.75, -0.5) {};
		\node [style=none] (8) at (-1.75, -0.5) {};
		\node [style=none] (9) at (-5.5, -0.25) {...};
		\node [style=none] (10) at (-8.45, -0.5) {...};
		\node [style=none] (11) at (-2.3, -0.5) {...};
		\node [style=none] (16) at (0.25, -0.25) {$=$};
		\node [style=none] (17) at (10.5, -0.25) {};
		\node [style=none] (18) at (2.25, -0.25) {};
		\node [style=none] (19) at (9.95, -0.25) {...};
		\node [style=none] (20) at (9.5, -0.25) {};
		\node [style=Z phase dot] (24) at (3.25, 0.75) {$\ \alpha_1\!\mp\!\frac\pi2\ $};
		\node [style=none] (26) at (3.25, -0.25) {};
		\node [style=none] (27) at (2.8, -0.25) {...};
		\node [style=Z phase dot] (30) at (9.5, 0.75) {$\ \alpha_n \!\mp\! \frac\pi2\ $};
		\node [style=none] (31) at (-7, -2) {};
		\node [style=Z phase dot] (32) at (-7, -1) {$\alpha_2$};
		\node [style=none] (33) at (-7.45, -2) {...};
		\node [style=none] (34) at (-8, -2) {};
		\node [style=Z phase dot] (35) at (-3.75, -1) {$\ \alpha_{n\!-\!1}\ $};
		\node [style=none] (36) at (-3.75, -2) {};
		\node [style=none] (37) at (-2.75, -2) {};
		\node [style=none] (38) at (-3.3, -2) {...};
		\node [style=Z phase dot] (39) at (4.5, -0.75) {$\ \alpha_2\!\mp\!\frac\pi2\ $};
		\node [style=none] (40) at (4.05, -1.75) {...};
		\node [style=none] (41) at (3.5, -1.75) {};
		\node [style=none] (42) at (4.5, -1.75) {};
		\node [style=Z phase dot] (43) at (8.25, -0.75) {$\ \alpha_{n\!-\!1} \!\mp\! \frac\pi2\ $};
		\node [style=none] (44) at (9.25, -1.75) {};
		\node [style=none] (45) at (8.7, -1.75) {...};
		\node [style=none] (46) at (8.25, -1.75) {};
		\node [style=none] (47) at (6.375, 0.25) {...};
		\node [style=none] (48) at (-5.5, 1.75) {\color{red} $*$};
		\node [style=none] (49) at (0.25, 0.5) {\lcomprule};
	\end{pgfonlayer}
	\begin{pgfonlayer}{edgelayer}
		\draw (5.center) to (2);
		\draw (6.center) to (2);
		\draw (7.center) to (4);
		\draw (8.center) to (4);
		\draw [style=hadamard edge] (2) to (0);
		\draw [style=hadamard edge] (4) to (0);
		\draw (18.center) to (24);
		\draw (26.center) to (24);
		\draw (20.center) to (30);
		\draw (17.center) to (30);
		\draw (34.center) to (32);
		\draw (31.center) to (32);
		\draw (36.center) to (35);
		\draw (37.center) to (35);
		\draw (41.center) to (39);
		\draw (42.center) to (39);
		\draw (46.center) to (43);
		\draw (44.center) to (43);
		\draw [style=hadamard edge] (24) to (30);
		\draw [style=hadamard edge] (30) to (43);
		\draw [style=hadamard edge] (43) to (39);
		\draw [style=hadamard edge] (39) to (24);
		\draw [style=hadamard edge] (24) to (43);
		\draw [style=hadamard edge] (39) to (30);
		\draw [style=hadamard edge] (0) to (32);
		\draw [style=hadamard edge] (0) to (35);
	\end{pgfonlayer}
\end{tikzpicture}
\end{equation}\vspace{0pt}

\begin{definition}
    The \emph{pivot} of a graph $G$ on an edge $(u,v)$ is the graph $G \wedge uv \:= G \star u \star v \star u$.
\end{definition}
Note that $G \wedge uv \equiv G \wedge vu$. This corresponds to exchanging the vertices $u$ and $v$, then complementing the edges between three vertex subsets: the common neighbourhood of $u, v$ ($A \:= N_G(u) \inter N_G(v)$), the exclusive neighbourhood of $u$ ($B \:= N_G(u) \setminus (N_G(v) \union \brace{v})$) and the exclusive neighbourhood of $v$ ($C \:= N_G(v) \setminus (N_G(u) \union \brace{u})$). This is illustrated below \cite{duncan2020graph}:

\begin{equation*}
    G \quad \begin{tikzpicture}
	\begin{pgfonlayer}{nodelayer}
		\node [style=vertex] (0) at (-2, 1.5) {};
		\node [style=vertex] (1) at (2, 1.5) {};
		\node [style=vertex set] (2) at (0, 0.5) {$~~A~~$};
		\node [style=vertex set] (3) at (-2.5, -0.75) {$~~B~~$};
		\node [style=vertex set] (4) at (2.5, -0.75) {$~~C~~$};
		\node [style=none] (5) at (2.75, 1.5) {$v$};
		\node [style=none] (6) at (-2.75, 1.5) {$u$};
		\node [style=none] (7) at (-0.25, 0.5) {};
		\node [style=none] (8) at (0, 0.25) {};
		\node [style=none] (9) at (0.25, 0.25) {};
		\node [style=none] (10) at (0.5, 0.5) {};
		\node [style=none] (11) at (-2.5, -0.5) {};
		\node [style=none] (12) at (-2.25, -0.75) {};
		\node [style=none] (13) at (-2.25, -0.75) {};
		\node [style=none] (14) at (-2.25, -1) {};
		\node [style=none] (15) at (2.25, -0.75) {};
		\node [style=none] (16) at (2.25, -1) {};
		\node [style=none] (17) at (2.25, -0.75) {};
		\node [style=none] (18) at (2.5, -0.5) {};
	\end{pgfonlayer}
	\begin{pgfonlayer}{edgelayer}
		\draw (0) to (1);
		\draw (0) to (3);
		\draw (2) to (1);
		\draw (0) to (2);
		\draw (1) to (4);
		\draw (7.center) to (11.center);
		\draw (13.center) to (8.center);
		\draw (9.center) to (15.center);
		\draw (10.center) to (18.center);
		\draw (17.center) to (12.center);
		\draw (14.center) to (16.center);
	\end{pgfonlayer}
\end{tikzpicture} \qquad\qquad\quad G \wedge uv \quad \begin{tikzpicture}
	\begin{pgfonlayer}{nodelayer}
		\node [style=vertex] (0) at (-2, 1.5) {};
		\node [style=vertex] (1) at (2, 1.5) {};
		\node [style=vertex set] (2) at (0, 0.5) {$~~A~~$};
		\node [style=vertex set] (3) at (-2.5, -1) {$~~B~~$};
		\node [style=vertex set] (4) at (2.25, -1) {$~~C~~$};
		\node [style=none] (5) at (-3, 1.5) {$v$};
		\node [style=none] (6) at (3, 1.5) {$u$};
		\node [style=none] (7) at (-2.75, -0.75) {};
		\node [style=none] (8) at (-2.25, -1.25) {};
		\node [style=none] (9) at (-0.25, 0.75) {};
		\node [style=none] (10) at (0.25, 0.25) {};
		\node [style=none] (11) at (-0.25, 0.25) {};
		\node [style=none] (12) at (0.25, 0.75) {};
		\node [style=none] (13) at (2.5, -0.75) {};
		\node [style=none] (14) at (2, -1.25) {};
		\node [style=none] (15) at (2.25, -1.25) {};
		\node [style=none] (16) at (2.25, -0.75) {};
		\node [style=none] (17) at (-2.5, -1.25) {};
		\node [style=none] (18) at (-2.5, -0.75) {};
	\end{pgfonlayer}
	\begin{pgfonlayer}{edgelayer}
		\draw (0) to (1);
		\draw (0) to (3);
		\draw (2) to (1);
		\draw (0) to (2);
		\draw (1) to (4);
		\draw (7.center) to (10.center);
		\draw (9.center) to (8.center);
		\draw (18.center) to (15.center);
		\draw (17.center) to (16.center);
		\draw (12.center) to (14.center);
		\draw (11.center) to (13.center);
	\end{pgfonlayer}
\end{tikzpicture}
\end{equation*}\vspace{0pt}

\noindent In the following rule \cite{duncan2020graph} parallel edges cancel, thus it is equivalent to performing a pivot of the underlying open graph on the marked vertices, removing them, and updating the phases as shown.

\begin{equation}
    \begin{tikzpicture}
	\begin{pgfonlayer}{nodelayer}
		\node [style=Z phase dot] (0) at (-9.25, 1.5) {$j\pi$};
		\node [style=Z phase dot] (2) at (-11, 1) {$\alpha_1$};
		\node [style=none] (16) at (0, 0) {$=$};
		\node [style=Z phase dot] (32) at (-11, -0.75) {$\alpha_n$};
		\node [style=Z phase dot] (37) at (-6.75, -2.25) {$\beta_1$};
		\node [style=Z phase dot] (41) at (-6.75, -0.5) {$\beta_n$};
		\node [style=Z phase dot] (45) at (-2.75, 1) {$\gamma_1$};
		\node [style=Z phase dot] (49) at (-2.75, -0.75) {$\gamma_n$};
		\node [style=Z phase dot] (51) at (-4.25, 1.5) {$k\pi$};
		\node [style=none, rotate=90] (52) at (-11, 0.125) {...};
		\node [style=none, rotate=90] (53) at (-6.75, -1.375) {...};
		\node [style=none, rotate=90] (54) at (-2.75, 0.125) {...};
		\node [style=Z phase dot] (55) at (3.25, 0.25) {$\ \alpha_n+k\pi\ $};
		\node [style=Z phase dot] (56) at (8, -1) {$\ \beta_n+(j+k+1)\pi\ $};
		\node [style=none, rotate=90] (57) at (8, -1.875) {...};
		\node [style=Z phase dot] (58) at (8, -2.75) {$\ \beta_1+(j+k+1)\pi\ $};
		\node [style=Z phase dot] (59) at (13, 2) {$\ \gamma_1+j\pi\ $};
		\node [style=Z phase dot] (60) at (3.25, 2) {$\ \alpha_1+k\pi\ $};
		\node [style=none, rotate=90] (61) at (13, 1.125) {...};
		\node [style=none, rotate=90] (62) at (3.25, 1.125) {...};
		\node [style=Z phase dot] (63) at (13, 0.25) {$\ \gamma_n+j\pi\ $};
		\node [style=none] (64) at (-10.5, 1.75) {};
		\node [style=none] (65) at (-11.5, 1.75) {};
		\node [style=none] (66) at (-11, 1.75) {...};
		\node [style=none] (67) at (-11.5, -1.5) {};
		\node [style=none] (68) at (-11, -1.5) {...};
		\node [style=none] (69) at (-10.5, -1.5) {};
		\node [style=none] (70) at (-7.25, -3) {};
		\node [style=none] (71) at (-6.75, -3) {...};
		\node [style=none] (72) at (-6.25, -3) {};
		\node [style=none] (73) at (-7.25, 0.25) {};
		\node [style=none] (74) at (-6.75, 0.25) {...};
		\node [style=none] (75) at (-6.25, 0.25) {};
		\node [style=none] (76) at (-3.25, 1.75) {};
		\node [style=none] (77) at (-2.75, 1.75) {...};
		\node [style=none] (78) at (-2.25, 1.75) {};
		\node [style=none] (79) at (-2.75, -1.5) {...};
		\node [style=none] (80) at (-3.25, -1.5) {};
		\node [style=none] (81) at (-2.25, -1.5) {};
		\node [style=none] (82) at (2.75, 2.75) {};
		\node [style=none] (83) at (3.25, 2.75) {...};
		\node [style=none] (84) at (3.75, 2.75) {};
		\node [style=none] (85) at (3.25, -0.5) {...};
		\node [style=none] (86) at (2.75, -0.5) {};
		\node [style=none] (87) at (3.75, -0.5) {};
		\node [style=none] (88) at (7.5, -0.25) {};
		\node [style=none] (89) at (8, -0.25) {...};
		\node [style=none] (90) at (8.5, -0.25) {};
		\node [style=none] (91) at (8, -3.5) {...};
		\node [style=none] (92) at (7.5, -3.5) {};
		\node [style=none] (93) at (8.5, -3.5) {};
		\node [style=none] (94) at (12.5, 2.75) {};
		\node [style=none] (95) at (13, 2.75) {...};
		\node [style=none] (96) at (13.5, 2.75) {};
		\node [style=none] (97) at (13, -0.5) {...};
		\node [style=none] (98) at (12.5, -0.5) {};
		\node [style=none] (99) at (13.5, -0.5) {};
		\node [style=none] (100) at (-9.25, 2.25) {\color{red} $*$};
		\node [style=none] (101) at (-4.25, 2.25) {\color{red} $*$};
		\node [style=none] (102) at (0, 0.75) {\pivotrule};
	\end{pgfonlayer}
	\begin{pgfonlayer}{edgelayer}
		\draw [style=hadamard edge] (2) to (0);
		\draw [style=hadamard edge] (0) to (32);
		\draw [style=hadamard edge] (0) to (51);
		\draw [style=hadamard edge] (0) to (37);
		\draw [style=hadamard edge] (0) to (41);
		\draw [style=hadamard edge] (51) to (37);
		\draw [style=hadamard edge] (51) to (41);
		\draw [style=hadamard edge] (51) to (45);
		\draw [style=hadamard edge] (51) to (49);
		\draw [style=hadamard edge] (60) to (56);
		\draw [style=hadamard edge] (60) to (58);
		\draw [style=hadamard edge] (55) to (56);
		\draw [style=hadamard edge] (55) to (58);
		\draw [style=hadamard edge] (59) to (56);
		\draw [style=hadamard edge] (59) to (58);
		\draw [style=hadamard edge] (56) to (63);
		\draw [style=hadamard edge] (63) to (58);
		\draw [style=hadamard edge] (60) to (59);
		\draw [style=hadamard edge] (60) to (63);
		\draw [style=hadamard edge] (55) to (63);
		\draw [style=hadamard edge] (55) to (59);
		\draw (64.center) to (2);
		\draw (2) to (65.center);
		\draw (67.center) to (32);
		\draw (69.center) to (32);
		\draw (70.center) to (37);
		\draw (72.center) to (37);
		\draw (41) to (75.center);
		\draw (41) to (73.center);
		\draw (45) to (78.center);
		\draw (45) to (76.center);
		\draw (49) to (80.center);
		\draw (49) to (81.center);
		\draw (60) to (82.center);
		\draw (60) to (84.center);
		\draw (59) to (94.center);
		\draw (59) to (96.center);
		\draw (98.center) to (63);
		\draw (99.center) to (63);
		\draw (86.center) to (55);
		\draw (87.center) to (55);
		\draw (56) to (88.center);
		\draw (56) to (90.center);
		\draw (92.center) to (58);
		\draw (93.center) to (58);
	\end{pgfonlayer}
\end{tikzpicture}
\end{equation}\vspace{0pt}

\subsection{Neighbour unfusion}\label{subsec:neighbour-unfusion}
Neighbour unfusion, the inverse of \IdFuseRule, alters a vertex's phase and neighbour set, enabling transformations on vertices with arbitrary phases. Previous papers have implemented specific cases of this rule, but here it is generalised to allow unfusing of any subset of neighbouring vertices $S_N \subset N(v)$:

\begin{equation}
    \begin{tikzpicture}
	\begin{pgfonlayer}{nodelayer}
		\node [style=Z phase dot] (2) at (3, 0) {$\beta$};
		\node [style=Z phase dot] (4) at (7, 0) {$\ \alpha - \beta\ $};
		\node [style=none] (5) at (2, 0.75) {};
		\node [style=none] (6) at (2, -0.75) {};
		\node [style=none, rotate=90] (10) at (2.05, 0) {...};
		\node [style=none] (16) at (0, 0) {$=$};
		\node [style=Z dot] (49) at (4.75, 0) {};
		\node [style=Z phase dot] (53) at (-5, 0) {$\alpha$};
		\node [style=none] (54) at (-4, 0.75) {};
		\node [style=none] (55) at (-4, -0.75) {};
		\node [style=none] (56) at (-6, -0.75) {};
		\node [style=none] (57) at (-6, 0.75) {};
		\node [style=none, rotate=90] (58) at (-6, 0) {...};
		\node [style=none, rotate=90] (59) at (-4, 0) {...};
		\node [style=none] (60) at (0, 0.75) {\neighbourunfuserule};
		\node [style=none] (65) at (8.5, 0.75) {};
		\node [style=none] (66) at (8.5, -0.75) {};
		\node [style=none, rotate=90] (67) at (8.5, 0) {...};
		\node [style=none] (68) at (7.5, 0.25) {};
		\node [style=none] (69) at (7.5, -0.25) {};
		\node [style=none] (70) at (-3, 0) {$\ \bigg\}\ S_N$};
		\node [style=none] (71) at (9.5, 0) {$\ \bigg\}\ S_N$};
	\end{pgfonlayer}
	\begin{pgfonlayer}{edgelayer}
		\draw (5.center) to (2);
		\draw (6.center) to (2);
		\draw [style=hadamard edge] (49) to (4);
		\draw [style=hadamard edge] (2) to (49);
		\draw (68.center) to (65.center);
		\draw (69.center) to (66.center);
		\draw (56.center) to (53);
		\draw (57.center) to (53);
		\draw (54.center) to (53);
		\draw (53) to (55.center);
	\end{pgfonlayer}
\end{tikzpicture}
\end{equation}\vspace{0pt}

\noindent In \refcite{duncan2020graph}, the case $|S_N| = 0$ is used to unfuse non-Clifford spiders into phase gadgets, thereby enabling \PivotRule on previously inapplicable vertices (see \cref{subsec:phase-teleportation}). The $|S_N| = 1$ case is also used for pivoting on an interior spider with an integer multiple of $\pi$ phase (\emph{Pauli spiders}) connected to a boundary spider, by unfusing the boundary spider first. The broader implementation of the $|S_N| = 1$ case was formalised as the neighbour unfusion rule in \refcite{staudacher2022reducing} for the purpose of 2Q-count reduction. \cref{subsec:nu-matches} demonstrates that allowing $|S_N| > 1$ is particularly impactful on 2Q-count reductions.

\subsection{Phase teleportation}\label{subsec:phase-teleportation}
Phase teleportation, introduced in \refcite{kissinger2019reducing}, reduces the T-count of a circuit without altering the diagram's structure. When a Pauli spider is adjacent to a non-Clifford spider, \NeighbourUnfuseRule is used to introduce a \emph{phase gadget}; a spider with arbitrary phase and a single Hadamard edge connection to a spider with no angle. Subsequently, \PivotRule is applied on the two marked vertices. If this adds a $\pi$ phase to the gadget's base, it is removable by \CopyRule:

\begin{equation}
    \begin{tikzpicture}
	\begin{pgfonlayer}{nodelayer}
		\node [style=Z phase dot] (0) at (-14, 0.25) {$j\pi$};
		\node [style=Z phase dot] (51) at (-11.5, 0.25) {$\alpha$};
		\node [style=none] (100) at (-6.5, 1) {\color{red} $*$};
		\node [style=none] (101) at (-4, 1) {\color{red} $*$};
		\node [style=none] (102) at (-12.25, -0.75) {};
		\node [style=none] (103) at (-10.75, -0.75) {};
		\node [style=none] (104) at (-13.25, -0.75) {};
		\node [style=none] (105) at (-14.75, -0.75) {};
		\node [style=none] (108) at (-9, 0) {$=$};
		\node [style=Z phase dot] (109) at (-6.5, 0.25) {$j\pi$};
		\node [style=none] (111) at (-4.75, -0.75) {};
		\node [style=none] (112) at (-3.25, -0.75) {};
		\node [style=none] (113) at (-7.25, -0.75) {};
		\node [style=none] (114) at (-5.75, -0.75) {};
		\node [style=Z dot] (117) at (-4, 0.25) {};
		\node [style=Z dot] (118) at (-3, 0.25) {};
		\node [style=Z phase dot] (119) at (-2, 0.25) {$\alpha$};
		\node [style=none] (120) at (-14, -0.75) {...};
		\node [style=none] (121) at (-11.5, -0.75) {...};
		\node [style=none] (122) at (-6.5, -0.75) {...};
		\node [style=none] (123) at (-4, -0.75) {...};
		\node [style=Z phase dot] (125) at (2, 0) {$\alpha$};
		\node [style=Z phase dot] (126) at (4, 0) {$\pi$};
		\node [style=none] (127) at (6.75, 0) {$=$};
		\node [style=none] (128) at (5, 0.75) {};
		\node [style=none] (129) at (5, -0.75) {};
		\node [style=none, rotate=90] (130) at (5, 0) {...};
		\node [style=Z phase dot] (131) at (8.5, 0) {$\ -\alpha\ $};
		\node [style=none] (133) at (11.5, 0.75) {};
		\node [style=none] (134) at (11.5, -0.75) {};
		\node [style=none, rotate=90] (135) at (11.5, 0) {...};
		\node [style=Z dot] (136) at (10.5, 0) {};
		\node [style=Z dot] (137) at (10.5, 0) {};
	\end{pgfonlayer}
	\begin{pgfonlayer}{edgelayer}
		\draw [style=hadamard edge] (0) to (51);
		\draw (104.center) to (0);
		\draw (105.center) to (0);
		\draw (51) to (102.center);
		\draw (51) to (103.center);
		\draw (113.center) to (109);
		\draw (114.center) to (109);
		\draw (117) to (112.center);
		\draw (117) to (111.center);
		\draw [style=hadamard edge] (109) to (117);
		\draw [style=hadamard edge] (119) to (118);
		\draw [style=hadamard edge] (118) to (117);
		\draw [style=hadamard edge] (125) to (126);
		\draw (126) to (128.center);
		\draw (126) to (129.center);
		\draw [style=hadamard edge] (131) to (137);
		\draw (137) to (133.center);
		\draw (137) to (134.center);
	\end{pgfonlayer}
\end{tikzpicture}
\end{equation}\vspace{0pt}

\noindent Two additional rules then reduce T-count by deleting phase gadgets. The first uses \IdFuseRule to remove a gadget with a single leg, and the second, denoted \emph{gadget-fusion}, fuses phase gadgets connected to the same set of spiders \cite{kissinger2019reducing}.

\begin{equation}
    \begin{tikzpicture}
	\begin{pgfonlayer}{nodelayer}
		\node [style=Z phase dot] (0) at (-11.5, 0) {$\alpha$};
		\node [style=none] (108) at (-5.5, 0) {$=$};
		\node [style=Z dot] (124) at (-10, 0) {};
		\node [style=Z phase dot] (125) at (-8.5, 0) {$\beta$};
		\node [style=none] (127) at (-7.5, 0.75) {};
		\node [style=none] (128) at (-7.5, -0.75) {};
		\node [style=none] (129) at (-7.5, 0.75) {};
		\node [style=none, rotate=90] (130) at (-7.5, 0) {...};
		\node [style=none] (131) at (-5.5, 0.75) {\idfuserule};
		\node [style=Z phase dot] (135) at (-3, 0) {$\ \alpha + \beta\ $};
		\node [style=none] (136) at (-1.75, 0.75) {};
		\node [style=none] (137) at (-1.75, -0.75) {};
		\node [style=none, rotate=90] (138) at (-1.75, 0) {...};
		\node [style=none] (139) at (-2.5, 0.25) {};
		\node [style=none] (140) at (-2.5, -0.25) {};
		\node [style=Z phase dot] (141) at (2.5, 0.75) {$\alpha_1$};
		\node [style=Z dot] (142) at (4, 0.75) {};
		\node [style=Z phase dot] (143) at (6, 1) {$\beta_1$};
		\node [style=none] (145) at (7, 0.25) {};
		\node [style=none] (146) at (7, 1.75) {};
		\node [style=none, rotate=90] (147) at (7, 1) {...};
		\node [style=Z phase dot] (148) at (2.5, -0.75) {$\alpha_2$};
		\node [style=Z dot] (149) at (4, -0.75) {};
		\node [style=Z phase dot] (150) at (6, -1) {$\beta_2$};
		\node [style=none] (152) at (7, -1.75) {};
		\node [style=none] (153) at (7, -0.25) {};
		\node [style=none, rotate=90] (154) at (7, -1) {...};
		\node [style=none, rotate=90] (155) at (6, 0) {...};
		\node [style=none] (156) at (9, 0) {$=$};
		\node [style=none] (157) at (9, 0.75) {\gadgetfusionrule};
		\node [style=Z phase dot] (160) at (15.25, 1) {$\beta_1$};
		\node [style=none] (161) at (16.25, 0.25) {};
		\node [style=none] (162) at (16.25, 1.75) {};
		\node [style=none, rotate=90] (163) at (16.25, 1) {...};
		\node [style=Z phase dot] (164) at (11.75, 0) {$\ \alpha_1+\alpha_2\ $};
		\node [style=Z dot] (165) at (14.25, 0) {};
		\node [style=Z phase dot] (166) at (15.25, -1) {$\beta_2$};
		\node [style=none] (167) at (16.25, -1.75) {};
		\node [style=none] (168) at (16.25, -0.25) {};
		\node [style=none, rotate=90] (169) at (16.25, -1) {...};
		\node [style=none, rotate=90] (170) at (15.25, 0) {...};
	\end{pgfonlayer}
	\begin{pgfonlayer}{edgelayer}
		\draw (125) to (129.center);
		\draw (125) to (128.center);
		\draw [style=hadamard edge] (124) to (125);
		\draw [style=hadamard edge] (0) to (124);
		\draw (139.center) to (136.center);
		\draw (140.center) to (137.center);
		\draw (143) to (146.center);
		\draw (143) to (145.center);
		\draw [style=hadamard edge] (142) to (143);
		\draw [style=hadamard edge] (141) to (142);
		\draw (150) to (153.center);
		\draw (150) to (152.center);
		\draw [style=hadamard edge] (149) to (150);
		\draw [style=hadamard edge] (148) to (149);
		\draw [style=hadamard edge] (149) to (143);
		\draw [style=hadamard edge] (142) to (150);
		\draw (160) to (162.center);
		\draw (160) to (161.center);
		\draw (166) to (168.center);
		\draw (166) to (167.center);
		\draw [style=hadamard edge] (165) to (166);
		\draw [style=hadamard edge] (164) to (165);
		\draw [style=hadamard edge] (165) to (160);
	\end{pgfonlayer}
\end{tikzpicture}
\end{equation}\vspace{0pt}

\noindent These rules, along with \IdFuseRule, \LCompRule, \PivotRule, can be applied iteratively to simplify a diagram. \LCompRule removes internal spiders with $\pm \pi/2$ phases (\emph{proper Clifford spiders}), while \PivotRule removes adjacent Pauli spider pairs. Pauli spiders adjacent to a boundary spider are also removed by first unfusing the boundary spider. Each rewrite removes a spider, hence the procedure terminates with the diagram in \emph{reduced gadget form}. Here every internal spider is non-Clifford or part of a non-Clifford phase gadget, and no more gadget simplifications can be applied. This strategy is implemented as \FullReduce in PyZX.

\emph{Phase teleportation} symbolically tracks non-Clifford phases as labelled phase variables $\alpha_i$, storing original phases in a table $i \smallmapsto \R$. Under \GadgetFusionRule, gadgets can also be unfused, combining phases onto one vertex. In \refcite{kissinger2019reducing} one variable stored the combined phase, and the other was immediately set to zero. We extend this here to track phases which would be stored on either variable. In the example below, gadgets are fused then immediately unfused. The sign of the variable is indicated by a multiplier $m_i \in \brace{-1,1}$, and any irrelevant Clifford phase is labelled $C$.

\begin{equation*}
    \begin{tikzpicture}
	\begin{pgfonlayer}{nodelayer}
		\node [style=Z phase dot] (141) at (-7.25, 0.625) {$\ m_i \alpha_i+ C\ $};
		\node [style=Z dot] (142) at (-5.25, 0.625) {};
		\node [style=Z phase dot] (143) at (-4.25, 0.875) {$\beta_i$};
		\node [style=none] (145) at (-3.5, 0.375) {};
		\node [style=none] (146) at (-3.5, 1.375) {};
		\node [style=none, rotate=90] (147) at (-3.5, 0.875) {...};
		\node [style=Z phase dot] (148) at (-7.25, -0.625) {$\ m_j \alpha_j+ C\ $};
		\node [style=Z dot] (149) at (-5.25, -0.625) {};
		\node [style=Z phase dot] (150) at (-4.25, -0.875) {$\beta_j$};
		\node [style=none] (152) at (-3.5, -1.375) {};
		\node [style=none] (153) at (-3.5, -0.375) {};
		\node [style=none, rotate=90] (154) at (-3.5, -0.875) {...};
		\node [style=none, rotate=90] (155) at (-4.25, 0) {...};
		\node [style=none] (162) at (-2.5, 0) {=};
		\node [style=Z phase dot] (163) at (2, 0.625) {$\ \delta_{\mu i} \left( m_i [\alpha_i+m_i m_j \alpha_j ] + C  \right) \ $};
		\node [style=Z dot] (164) at (6.125, 0.625) {};
		\node [style=Z phase dot] (165) at (7.125, 0.875) {$\beta_i$};
		\node [style=none] (166) at (7.875, 0.375) {};
		\node [style=none] (167) at (7.875, 1.375) {};
		\node [style=none, rotate=90] (168) at (7.875, 0.875) {...};
		\node [style=Z phase dot] (169) at (2, -0.625) {$\ \delta_{\mu j} \left( m_j [\alpha_j+m_i m_j \alpha_i ]  + C \right) \ $};
		\node [style=Z dot] (170) at (6.125, -0.625) {};
		\node [style=Z phase dot] (171) at (7.125, -0.875) {$\beta_j$};
		\node [style=none] (172) at (7.875, -1.375) {};
		\node [style=none] (173) at (7.875, -0.375) {};
		\node [style=none, rotate=90] (174) at (7.875, -0.875) {...};
		\node [style=none, rotate=90] (175) at (7.125, 0) {...};
	\end{pgfonlayer}
	\begin{pgfonlayer}{edgelayer}
		\draw (143) to (146.center);
		\draw (143) to (145.center);
		\draw [style=hadamard edge] (142) to (143);
		\draw [style=hadamard edge, in=180, out=0] (141) to (142);
		\draw (150) to (153.center);
		\draw (150) to (152.center);
		\draw [style=hadamard edge] (149) to (150);
		\draw [style=hadamard edge] (148) to (149);
		\draw [style=hadamard edge] (149) to (143);
		\draw [style=hadamard edge] (142) to (150);
		\draw (165) to (167.center);
		\draw (165) to (166.center);
		\draw [style=hadamard edge] (164) to (165);
		\draw [style=hadamard edge] (163) to (164);
		\draw (171) to (173.center);
		\draw (171) to (172.center);
		\draw [style=hadamard edge] (170) to (171);
		\draw [style=hadamard edge] (169) to (170);
		\draw [style=hadamard edge] (170) to (165);
		\draw [style=hadamard edge] (164) to (171);
	\end{pgfonlayer}
\end{tikzpicture}
\end{equation*}

\noindent Here $\delta_{\mu i}$ is the Kronecker delta, with $\mu \in \brace{i,j}$. Multipliers are factorised for consistent backtracking, using $(m_i)^2 = 1$. Delaying gadget fusions does not impact other transformations. This is extendable to any number of gadgets.
\begin{lemma}
    Given phase variables $\brace{\alpha_1 \cdots \alpha_n}$ fusing together, each variable can be transformed as follows:
    \vspace{-5pt}\begin{equation*}
        \alpha_i \; \mapsto \;\; \delta_{\lambda i} \; m_i \sum_{j=1}^{n} m_j \alpha_j
        \vspace{-5pt}\end{equation*}
    where $\lambda \in \brace{1 \cdots n}$ is an index to be chosen, and $m_i$ are multipliers of the respective gadgets.
\end{lemma}
Each potential phase differs at most by a sign flip, rendering the choice of $\lambda$ inconsequential to the T-count. However, delaying the choice opens a larger space of potential transformations, expanded further by noting that remaining phase variables can be fixed to any value.
\begin{lemma}\label{lem:phase-fixing}
    Given phase variables $\brace{\alpha_1 \cdots \alpha_n}$ fusing together and a variable $\alpha_j \in \brace{\alpha_1 \cdots \alpha_n}$, phase fixing sets $\alpha_\kappa = \beta$, with the remaining variables being updated as follows:
    \vspace{-5pt}\begin{equation*}
        \alpha_{i \neq \kappa} \; \mapsto \;\; \delta_{\sigma i} \; m_i \, \biggl( \Bigl( \; \sum_{j=1}^{n} m_j \alpha_j \; \Bigr) - m_\kappa \beta \biggr)
        \vspace{-5pt}\end{equation*}
    where $\sigma \in \brace{1 \cdots n} \setminus \brace{\kappa}$ is an index to be chosen, and $m_i$ are the respective multipliers.
\end{lemma}
We denote the approach of postponing phase fixing, allowing phases to teleport around throughout later simplifications, as \emph{phase jumping}. Notably, the specific sequence of simplifications leading to reduced gadget form never impacts the resulting T-count.
\section{Circuit extraction and flow}\label{sec:extraction-and-flow}
In the one-way model of measurement based quantum computing \cite{raussendorf2001one}, generalised flow (\emph{gflow}) is a necessary property for a deterministic graph-state \cite{browne2007generalized}. Whilst circuit extraction for general ZX-diagrams is \#P-Hard \cite{de2022circuit}, there is a polynomial time algorithm for diagrams with an underlying open graph admitting a gflow \cite{backens2021there}. An exact calculation of the gflow is not required; the knowledge that one exists suffices. \refcite{duncan2020graph} showed that \IdFuseRule, \LCompRule and \PivotRule always preserve the existence of a gflow. Furthermore, applying \NeighbourUnfuseRule with either $|S_N|=0$ or $|S_N|=1$, where $S_N \subseteq (I \union O)$, also preserves gflow. Thus, these transformations can be applied iteratively with confidence that a circuit can be extracted. This is powerful for diagrams containing solely Clifford spiders, but otherwise often results in a higher 2Q-count than the original (see \cref{fig:gen-results}). Causal flow (\emph{cflow}) is a stricter condition than gflow which is sufficient, though not necessary, for a deterministic MBQC graph-state \cite{danos2006determinism}.
\begin{definition}\label{def:partial-order}
    A relation $\preceq$ is a \emph{partial order} on a set $V$ if it has \textit{reflexivity} $\left( u \preceq u\ \;\;\forall u \in V \right)$, \textit{antisymmetry} $\left( u \preceq v \;\;\textit{and}\;\; v \preceq u \implies u = v \right)$ and \textit{transitivity} $\left( u \preceq v \;\;\textit{and}\;\; v \preceq w \implies u \preceq w \right)$.
\end{definition}
\begin{definition}\label{def:causal-flow}
    A \emph{cflow} $(f,\preceq)$ on an open graph $(G,I,O)$ consists of a \emph{successor function} $f: \bar{O} \smallmapsto \bar{I}$ and a partial order $\preceq$ over $V(G)$\,, such that $\forall \; u \in \bar{O}$ and $\forall \; v \in V(G)$ the following conditions are satisfied:
    \begin{align*}
        \condi \quad u \sim f(u)     &  &
        \condii \quad u \preceq f(u) &  &
        \condiii \quad v \sim f(u) \;\implies\; u \preceq v
    \end{align*}
\end{definition}
The partial order of a cflow can also be derived from successor function as its natural pre-order.
\begin{definition}\label{def:natural-pre-order}
    For a graph $G$ and successor function $f: V(G) \smallmapsto V(G)$, the \emph{natural pre-order $\preceq_f$} for $f$ is the transitive closure on $V(G)$ of the following $\forall u,v \in V(G)$:
    \begin{align*}
        \condi \quad u \preceq_f u     &  &
        \condii \quad u \preceq_f f(u) &  &
        \condiii \quad v \sim f(u) \;\implies\; u \preceq_f v
    \end{align*}
\end{definition}
\begin{theorem}[\citenum{de2008finding}]\label{thm:cflow-pre-order-is-partial}
    For an open graph $(G,I,O)$ with successor function $f$, $(f,\preceq_f)$ is a cflow iff the natural pre-order $\preceq_f$ is a partial order over $V(G)$.
\end{theorem}
The sequence of vertices defined iteratively by the successor function of a cflow form a directed path (dipath) from an input to an output vertex.
\begin{theorem}[\citenum{de2008finding}]\label{thm:unique-IO-dipaths}
    For an open graph $(G,I,O)$ with $|I|=|O|$ and cflow $(f,\preceq)$, the collection of dipaths defined by the successor function $f$ is the only maximum collection of vertex-disjoint \IO dipaths.
\end{theorem}
The existence of a maximum collection of vertex-disjoint \IO dipaths whose successor function does not define a cflow, implies that the open graph does not admit cflow. A useful way to gain intuition of the cflow conditions is to examine the successor function's \emph{influencing digraph}.
\begin{definition}\label{def:influencing-digraph}
    The \emph{influencing digraph} $\Gamma(V,A)$ of a successor function $f$ on an open graph $(G,I,O)$ consists of vertices $V(\Gamma) \:= V(G)$ and directed edges (arcs) defined by:
    \begin{align*}
        A(\Gamma) \:= \set{(u \smallto v)}{u \in \bar{O} \;,\; v \in \brace{f(u)} \union \brace{N_G(f(u))\setminus\brace{u}}}
    \end{align*}
\end{definition}
\begin{theorem}[\citenum{de2008finding}]\label{thm:no-cycles}
    For an open graph $(G,I,O)$ with $|I|=|O|$ and a successor function $f$ defining a maximum collection of vertex-disjoint dipaths, $f$ defines a causal flow iff the influencing digraph of $f$ on $(G,I,O)$ is acyclic.
\end{theorem}
This can be interpreted as representing the temporal flow of information through the graph, with the cflow conditions prohibiting causal loops. Cflow-admitting diagrams, in contrast to ones with gflow, directly translate to a circuit representation. In such diagrams, \emph{flow dipaths correspond to qubit lines} \cite{fagan2019optimising}. Edges along dipaths are extracted to Hadamard gates, phase spiders to phase gates, and edges between dipaths to CZ gates. Thus, a graph-like ZX-diagram with cflow essentially represents an alternative circuit form, using the gate set $\brace{CZ,Z_\alpha,H}$.

An important consequence of this is the ability to directly compute the 2Q-count of the extracted circuit from the ZX-diagram. As qubit paths are disjoint, each non-input vertex is the target of exactly one successor function mapping. Consequently, the number of directed edges equals $|V| - |I|$, leading to the following:
\begin{lemma}\label{lem:N2Q}
    In a graph-like ZX-diagram with an underlying open graph $(G,I,O)$ that admits a cflow, the number of two-qubit gates in the equivalent extracted circuit is:
    \begin{equation*}
        N_{2Q} = |E| - |V| + |I|
    \end{equation*}
\end{lemma}
Note that \cref{def:causal-flow} is only valid for graph-like ZX-diagrams without phase gadgets. In \cref{app:flow-extension} cflow is extended to include diagrams with phase gadgets. These can be extracted using phase polynomial synthesis \cite{de2020architecture}. However, predicting the impact of integrating a phase polynomial in the extracted circuit is complex, as interactions with later-added gadgets may occur. Hence, phase gadget unfusions are excluded from the presented algorithm.
\section{Flow preserving transformations}\label{sec:flow-preservation}
To extract a circuit, it is crucial that a transformation maintains the existence of flow. Both gflow and cflow are considered due to their distinct advantages and disadvantages.

Preserving gflow allows the diagram to deviate from a strict circuit-like structure, potentially enabling optimisations not possible in circuits. However, in this case, the calculation of $N_{2Q}$ as in \cref{lem:N2Q} is only an approximate heuristic, yielding an unpredictable 2Q-count. This also limits scope for more advanced decision strategies. Another advantage of maintaining cflow is that a circuit-like structure provides direct visibility of the qubits and their connectivity. This opens up possiblities for optimisations when connectivity for a fixed circuit architecture is required, as well as allowing translations of circuit rewrite rules into transformations on graph-like ZX-diagrams, conditional on the cflow.

\begin{figure}[ht]
    \centering
    \begin{subfigure}[b]{0.4\textwidth}
        \centering
        \includegraphics[width=\textwidth]{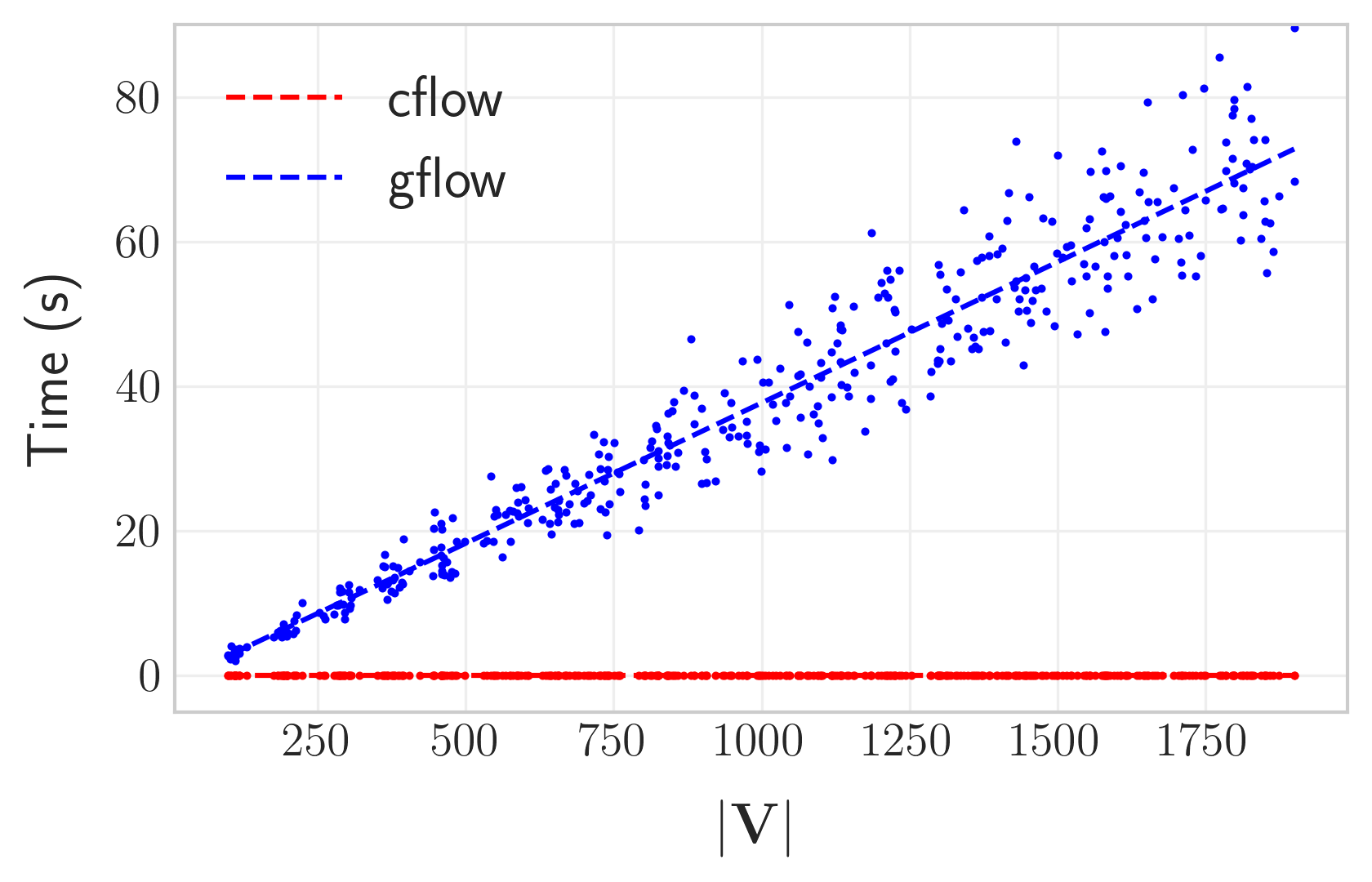}
    \end{subfigure}
    \hspace{20pt}
    \begin{subfigure}[b]{0.4\textwidth}
        \centering
        \includegraphics[width=\textwidth]{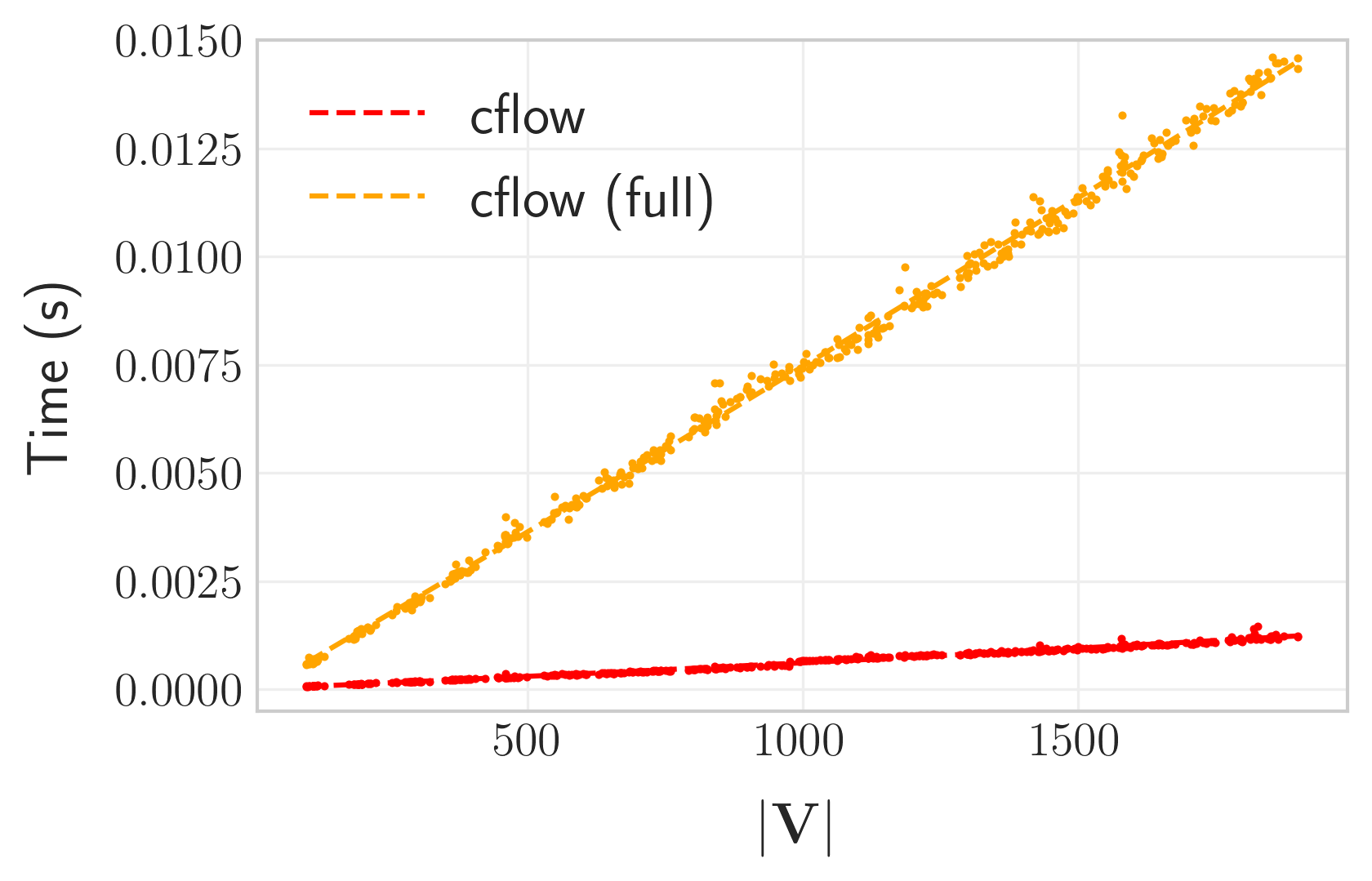}
    \end{subfigure}
    \caption{Plots contrasting time taken by flow-finding algorithms as a function of the number of vertices. (Left) compares the cflow and gflow algorithms of \refcite{mhalla2008finding}. (Right) constrasts the same cflow algorithm against the full cflow algorithm of \refcite{de2006complete}.}
    \label{fig:cflow-algs}
\end{figure}

To date, \IdFuseRule, \LCompRule, \PivotRule and some \NeighbourUnfuseRule applications are known to preserve gflow. The preservation criteria for general \NeighbourUnfuseRule applications remains an active research area. For cflow, \emph{only} \IdFuseRule is known to ensure preservation. When a transformation's flow preservation is uncertain, it requires validation to confirm that the resulting diagram has flow. The most general way to do this is to apply the transformation and check for flow, discarding non-preserving candidates. \refcite{mhalla2008finding} introduces efficient algorithms for both cflow and gflow. The cflow algorithm is more computationally feasible ($O(|I||V|)$) compared to the gflow algorithm ($O(|V|^4)$), as illustrated in \cref{fig:cflow-algs}. Thus, the results presented in \cref{subsec:results} do not consider gflow preservation including \NeighbourUnfuseRule, as efficient application necessitates well-defined preservation criteria.

\subsection{Local causal flow preservation}
Understanding general cflow preservation can be helped by exploring \emph{local cflow preserving transformations}, where the successor function changes only in the vicinity of the transformation. This leads to sufficient conditions for cflow preservation, potentially beneficial for large circuits where near-optimal results are acceptable for faster runtime.

\begin{definition}\label{def:open-subgraph-for-f}
    For an open graph $(G,I,O)$ with successor function $f$ and a vertex subset $S \subseteq V(G)$, the \emph{open subgraph $(G_s,I_s,O_s)$ of $S$ for $f$} has $V(G_s) \:= S$, $E(G_s) \:= \set{(u,v) \in E(G)}{\; u,v \in S}$ and input/output sets defined as follows:
    \begin{align*}
        I_s & \:= \set{v \in S}{(\;\exists \; u \in V \setminus S \st f(u)=v \;) \union (v\in I)} \\
        O_s & \:= \set{v \in S}{(f(v) \in V \setminus S) \union (v \in O) }
    \end{align*}
\end{definition}

\begin{definition}\label{def:inf-subdigraph}
    For an open graph $(G,I,O)$ with successor function $f$, influencing digraph $\Gamma$ and vertex subset $S \subseteq V(G)$, the \emph{influencing subdigraph of $S$ for $f$}, $\Gamma_s$ has vertices $V(\Gamma_s) \:= S$ and arcs $A(\Gamma_s) \:= \set{(u \smallto v) \in \Gamma}{u, v \in  S}$.
\end{definition}

\cref{def:open-subgraph-for-f} ensures that the subgraph's successor function aligns with the overarching graph. The subdigraph encompasses all dipaths between vertices in $S$.

\begin{theorem}\label{thm:open-subgraph-cflow}
    If an open graph $(G,I,O)$ admits a cflow $(f,\preceq)$, then any open subgraph of a vertex subset $S \subseteq V(G)$ also admits a cflow with the same successor function $f$.
\end{theorem}
\begin{proof}
    By \cref{thm:no-cycles}, the influencing digraph $\Gamma$ of $f$ on $(G,I,O)$ is acyclic. \cref{def:open-subgraph-for-f} ensures $\bar{O_s}\subseteq \bar{O}$ and $\forall u \in \bar{O_s}$, $\brace{N_{G_s}(f(u))} \subseteq \brace{N_G(f(u))}$. Thus, due to \cref{def:influencing-digraph}, the influencing digraph $\Gamma_s$ of $f$ on $(G_s,I_s,O_s)$ only contains arcs $A(\Gamma_s)\subseteq A(\Gamma)$, which are also acyclic. Since $f$ defines a maximum collection of vertex-disjoint dipaths in $(G_s,I_s,O_s)$, it constitutes a cflow by \cref{thm:no-cycles}.
\end{proof}

Next, consider a \emph{general} transformation $\mT$, where the transformation vertex set encompasses those with modified neighbour sets.

\begin{definition}\label{def:transformatation-vertex-set}
    For an open graph $(G,I,O)$ with cflow $(f, \preceq)$ and transformation $\mT: G \smallto G'$, the \emph{transformation vertex set} $V_\mT$ and its image $V'_\mT$ are defined as $S_\mT \inter V(G)$ and $S_\mT \inter V(G')$, defining $S_\mT$ as:
    \begin{align*}
        S_{\mT} & \:= \set{v \in V(G') \union V(G)}{\;\exists \; u \in V(G) \union V(G') \st (v,u) \in E(G) \symdiff E(G')}
    \end{align*}
    where $\symdiff$ is the symmetric set difference, i.e. $A \symdiff B \:= (A \union B) \setminus (A \inter B)$.
\end{definition}

A set of sufficient conditons for local cflow preservation within the open subgraph of the image of the transformation vertex set can then be defined. Here we want to avoid recalculating transitive closures over the entire graph.

\begin{definition}\label{def:extended-neighbourhood}
    For an open graph $G$ and vertex subset $S \in V(G)$, the \emph{extended neighbourhood of $S$} is defined as $N(S) \:= S \union \set{v \in V(G)}{\exists u \in S \st u \sim v}$.
\end{definition}

\begin{definition}\label{def:natural-preorder-subrelation}
    For a graph $G$ with successor function $f$, the \emph{natural pre-order subrelation} for a vertex subset $S \subseteq V(G)$ is defined by the following, where $\preceq_f$ is the natural pre-order for $f$ over $V(G)$:
    \begin{align*}
        \mR[S,G,f] \:= \brace{(u \smallto v) \st u,v \in S \;,\; u \neq v \;,\; u \preceq_f v}
    \end{align*}
\end{definition}

\begin{theorem}
    For an open graph $(G,I,O)$ with cflow $(f, \preceq)$ undergoing a transformation $\mT: (G,I,O) \smallto (G',I,O)$, the following conditions are sufficient for $(G',I,O)$ to admit a cflow defined by the successor function:
    \begin{equation*}
        f'(u) \:=
        \begin{cases}
            f'_\mT(u) & \forall \; u \in V(G'_\mT) \setminus O_n                                   \\
            f(u)      & \forall \; u \in (V(G') \setminus (V(G'_\mT) \setminus O_\mT)) \setminus O
        \end{cases}
    \end{equation*}
    \begin{romanum}
        \item The open subgraph $(G'_\mT, I_\mT, O_\mT)$ of the image of the transformation vertex set $V'_\mT$ for $f$ has a maximum collection of vertex-disjoint \IO dipaths, defined by a successor function $f'_\mT$.
        \item The influencing subdigraph $\Gamma'_{N(\mT)}$ of the extended neighbourhood of $V'_\mT$, $N(V'_\mT)$, for $f'$ is acyclic.
        \item The transitive closure of $\mR[N(V'_\mT) \setminus V'_\mT, G'_{N(\mT)}, f'] \union \mR[N(V'_\mT) \setminus V'_\mT, G, f]$, where $G'_{N(\mT)}$ is the subgraph of $N(V'_\mT)$, is acyclic.
    \end{romanum}
\end{theorem}
\begin{proof}
    Following \cref{thm:unique-IO-dipaths}, determining if a graph has cflow is equivalent to finding a maximum collection of vertex-disjoint \IO dipaths and checking if the associated influencing digraph is acyclic. \cref{def:open-subgraph-for-f} ensures that meeting \condi implies that a maximal collection of vertex-disjoint dipaths, defined by $f'$, exists for $(G',I,O)$. In this case cflow is either locally preserved or \cref{thm:unique-IO-dipaths} is contradicted and cflow is not preserved at all.

    Potential new arcs ouside $G'_\mT$ arise only from vertices in $V'_\mT \setminus O_\mT$ to those in $N(V'_\mT) \setminus V'_\mT$, confining digraph changes and potential cycles to $\Gamma'_{N(\mT)}$. Cycles must either be completely within $V'_\mT$ or include vertices from both $V'_\mT$ and $N(V'_\mT) \setminus V'_\mT$.

    \noindent\begin{minipage}{\hsize}
        \noindent\begin{minipage}{0.69\hsize}
            \setlength{\parskip}{5pt}\setlength{\parindent}{15pt}
            Two variations of such a cycle are possible. The first, contained within $G'_\mT$ (cycle 1 in \cref{fig:local-flow-cycles}), is checked by verifying the acyclicity of the transitive closure of $\Gamma'_{N(\mT)}$ for $f'$. If such a cycle exists, then this maximum collection of vertex-disjoint dipaths does not define a cflow, and therefore is not preserved (locally or globally). Hence \condii is necessary provided \condi is fufilled.

            The other type involves cycles passing through at least one vertex in $V'_\mT$, two $u,v \in N(V'_\mT) \setminus V'_\mT$ and at least one in $V(G') \setminus N(V'_\mT)$ (cycle 2 in \cref{fig:local-flow-cycles}). This could also loop in and out of the subgraph (cycle 3 in \cref{fig:local-flow-cycles}). To avoid recalculating a transitive closure over the entire graph, \condiii is sufficient for the absense of such cycles. $\mR[N(V'_\mT) \setminus V'_\mT, G'_{N(\mT)}, f']$ is the relation between all vertices $u,v \in N(V'_\mT) \setminus V'_\mT$ arising due to changes in the successor function. $\mR[N(V'_\mT) \setminus V'_\mT, G, f]$ is the relation between these vertices in the original subgraph. If a transitive closure between these two is acyclic, then the influencing digraph of $f'$ on $(G',I,O)$ is also acyclic, confirming a cflow.
        \end{minipage}%
        \hfill
        \begin{minipage}{0.27\hsize}
            \centering
            \begin{tikzpicture}
	\begin{pgfonlayer}{nodelayer}
		\node [style=none, color=violet] (8) at (-1.625, 7) {\footnotesize$G'_{N(\mathcal{T})}$};
		\node [style=none, color=violet] (9) at (-1, 4.75) {\footnotesize $G'_\mathcal{T}$};
		\node [style=none] (10) at (-3.25, 4) {};
		\node [style=none] (11) at (-3.25, 6.25) {};
		\node [style=none] (12) at (-3.25, 7.75) {};
		\node [style=none] (13) at (-3.25, 1.75) {};
		\node [style=none] (14) at (-3.25, 0.25) {};
		\node [style=none] (15) at (5.75, 4) {};
		\node [style=none] (16) at (5.75, 6.25) {};
		\node [style=none] (17) at (5.75, 7.75) {};
		\node [style=none] (18) at (5.75, 1.75) {};
		\node [style=none] (19) at (5.75, 0.25) {};
		\node [style=none] (20) at (4, 4) {};
		\node [style=none] (21) at (3.125, 6.25) {};
		\node [style=none] (23) at (3.15, 1.75) {};
		\node [style=none] (24) at (2.5, 4) {};
		\node [style=none] (25) at (-2, 4) {};
		\node [style=none] (26) at (-1.175, 6.25) {};
		\node [style=none] (27) at (-0.5, 4) {};
		\node [style=none] (28) at (-1.15, 1.75) {};
		\node [style=none, rotate=90] (29) at (-3.75, 6.25) {...};
		\node [style=none] (30) at (0, 6.5) {};
		\node [style=none] (31) at (0.75, 5) {};
		\node [style=none, rotate=90] (32) at (-3.75, 7.75) {...};
		\node [style=none, rotate=90] (33) at (-3.75, 4) {...};
		\node [style=none, rotate=90] (34) at (-3.75, 1.75) {...};
		\node [style=none, rotate=90] (35) at (-3.75, 0.25) {...};
		\node [style=none] (36) at (-1.25, 1) {};
		\node [style=none] (37) at (0.5, 3.25) {};
		\node [style=none] (38) at (-0.5, 4) {};
		\node [style=none] (39) at (2.5, 4) {};
		\node [style=none] (40) at (-2, 4) {};
		\node [style=none] (41) at (4, 4) {};
		\node [style=none] (42) at (1.75, 4) {};
		\node [style=none] (43) at (1.25, 4) {};
		\node [style=none] (44) at (5, 2.5) {};
		\node [style=none] (45) at (5, 5.5) {};
		\node [style=none] (46) at (2.75, 5) {};
		\node [style=none] (47) at (2.75, 3) {};
		\node [style=none] (48) at (4.5, 4.75) {};
		\node [style=none] (49) at (4.5, 3.25) {};
		\node [style=none] (51) at (2.25, 2.25) {};
		\node [style=none] (52) at (2.25, 5.75) {};
		\node [style=none] (53) at (0.375, 5.8) {\footnotesize 1};
		\node [style=none] (54) at (-0.375, 2.2) {\footnotesize 2};
		\node [style=none] (55) at (4.325, 5.325) {\footnotesize 3};
	\end{pgfonlayer}
	\begin{pgfonlayer}{edgelayer}
		\draw [style=dir hadamard edge] (12.center) to (17.center);
		\draw [style=dir hadamard edge] (11.center) to (26.center);
		\draw [style=dir hadamard edge] (26.center) to (21.center);
		\draw [style=dir hadamard edge] (21.center) to (16.center);
		\draw [style=dir hadamard edge] (10.center) to (25.center);
		\draw [style=dir hadamard edge] (25.center) to (27.center);
		\draw [style=dir hadamard edge] (27.center) to (24.center);
		\draw [style=dir hadamard edge] (24.center) to (20.center);
		\draw [style=dir hadamard edge] (20.center) to (15.center);
		\draw [style=dir hadamard edge] (13.center) to (28.center);
		\draw [style=dir hadamard edge] (28.center) to (23.center);
		\draw [style=dir hadamard edge] (23.center) to (18.center);
		\draw [style=dir hadamard edge] (14.center) to (19.center);
		\draw [style=diredge, bend right=90, looseness=1.75] (30.center) to (31.center);
		\draw [style=diredge, bend right=90, looseness=1.75] (31.center) to (30.center);
		\draw [style=diredge, bend right=90, looseness=1.75] (36.center) to (37.center);
		\draw [style=diredge, bend right=90, looseness=1.75] (37.center) to (36.center);
		\draw [color=violet, bend left=90, looseness=1.75] (38.center) to (39.center);
		\draw [color=violet, bend right=90, looseness=1.75] (38.center) to (39.center);
		\draw [color=violet, bend left=90, looseness=1.75] (40.center) to (41.center);
		\draw [color=violet, bend right=90, looseness=1.75] (40.center) to (41.center);
		\draw [style=diredge, bend left=45] (47.center) to (42.center);
		\draw [style=diredge, bend left=45] (42.center) to (46.center);
		\draw [bend left=15] (46.center) to (48.center);
		\draw [bend right=90, looseness=1.75] (44.center) to (49.center);
		\draw [style=diredge, bend left=15] (49.center) to (47.center);
		\draw [style=diredge, bend right=90, looseness=1.75] (48.center) to (45.center);
		\draw [bend right] (45.center) to (52.center);
		\draw [bend right] (43.center) to (51.center);
		\draw [style=diredge, bend right] (51.center) to (44.center);
		\draw [style=diredge, bend right] (52.center) to (43.center);
	\end{pgfonlayer}
\end{tikzpicture}
            \captionsetup{hypcap=false}
            
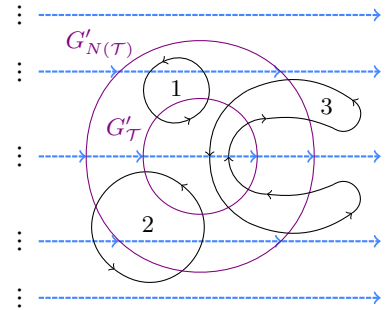
\captionof{figure}{Illustration of the three possible cycles in the influencing digraph which are not contained within the vertex set $V'_\mT$.}\label{fig:local-flow-cycles}
        \end{minipage}%
    \end{minipage}\vspace{-\intextsep}
\end{proof}
Condition \condiii is sufficient, but not necessary, for local cflow preservation. Originally, all dipaths from $v \smallto u$ might pass through the influencing subdigraph $\Gamma_{\mT}$ of $V_\mT$. Post-transformation, new $u \smallto v$ dipaths could then arise without forming a cycle. Verifying this, however,  requires either recalculating the entire graph's natural pre-order or storing the dipaths, both being computationally expensive. Instead, \condiii only requires calculating the transitive closure for two relations sets within a small vertex subset of vertices adjacent to $V'_\mT$. While the algorithm from \refcite{mhalla2008finding} returns the successor function and a labelling $l(v)$ ($u \preceq v \;\textiff\; l(u) > l(v)$), it does not determine the partial order between all vertex pairs. The algorithm from \refcite{de2006complete} can be used for a full order, but has a higher computational complexity of $O(|I|^2|V|)$ (shown in \cref{fig:cflow-algs}).

\noindent\begin{minipage}{\hsize}
    \noindent\begin{minipage}{0.75\hsize}
        \setlength{\parskip}{5pt}\setlength{\parindent}{15pt}
        Local flow preservation, though sufficient, is not necessary for preserving cflow. If conditon \condi is not satisifed, the graph might still admit cflow with a different successor function external to $G'_\mT$, as illustrated in \cref{fig:global-flow-change}. Here the altered input/output sets allow for new dipaths. Defining precise conditions for this scenario is an area for future research. It may be that a specific transformation is guaranteed to have a maximum collection of dipaths, hence cflow can only be locally preserved. The most reliable cflow verification method is recalculating cflow for the entire transformed graph.
    \end{minipage}%
    \hspace{15pt}
    \begin{minipage}{0.185\hsize}
        \centering
        \begin{tikzpicture}
	\begin{pgfonlayer}{nodelayer}
		\node [style=none, color=violet] (9) at (-1.25, 1) {\footnotesize $G'_\mathcal{T}$};
		\node [style=none] (11) at (-2.25, 0.25) {};
		\node [style=none] (12) at (-2.25, 2.25) {};
		\node [style=none] (13) at (-2.25, -1.75) {};
		\node [style=none] (16) at (3.75, 0.25) {};
		\node [style=none] (17) at (3.75, 2.25) {};
		\node [style=none] (18) at (3.75, -1.75) {};
		\node [style=none, rotate=90] (29) at (-2.75, 0.25) {...};
		\node [style=none, rotate=90] (32) at (-2.75, 2.25) {...};
		\node [style=none, rotate=90] (34) at (-2.75, -1.75) {...};
		\node [style=none] (38) at (-0.75, 0.25) {};
		\node [style=none] (39) at (2.25, 0.25) {};
		\node [style=none] (42) at (-1.5, 2.25) {};
		\node [style=none] (43) at (3, 2.25) {};
		\node [style=none] (45) at (1.5, 1.65) {};
		\node [style=none] (46) at (0, 1.65) {};
		\node [style=none] (47) at (0, 0.25) {};
		\node [style=none] (48) at (1.5, 0.25) {};
	\end{pgfonlayer}
	\begin{pgfonlayer}{edgelayer}
		\draw [color=violet, bend left=90, looseness=1.75] (38.center) to (39.center);
		\draw [color=violet, bend right=90, looseness=1.75] (38.center) to (39.center);
		\draw [style=dir hadamard edge] (12.center) to (42.center);
		\draw [style=dir hadamard edge] (43.center) to (17.center);
		\draw [style=dir hadamard edge] (47.center) to (45.center);
		\draw [style=dir hadamard edge] (45.center) to (43.center);
		\draw [style=dir hadamard edge] (42.center) to (46.center);
		\draw [style=dir hadamard edge] (46.center) to (48.center);
		\draw [style=dir hadamard edge] (48.center) to (39.center);
		\draw [style=dir hadamard edge] (38.center) to (47.center);
		\draw [style=dir hadamard edge] (11.center) to (38.center);
		\draw [style=dir hadamard edge] (39.center) to (16.center);
		\draw [style=dir hadamard edge] (13.center) to (18.center);
	\end{pgfonlayer}
\end{tikzpicture}
        \captionsetup{hypcap=false}
        
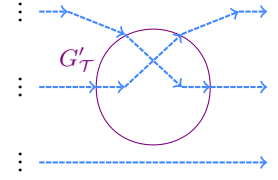
\captionof{figure}{Illustration of non-local cflow change.}\label{fig:global-flow-change}
    \end{minipage}%
    \hfill
\end{minipage}

The above theorem is also valid when replacing $V'_\mT$ with any larger subset of vertices $S$ such that $V'_\mT \subseteq S$. Larger subsets are more likely to include valid changes to the successor function.

Given the linear scaling of the cflow algorithm in \refcite{mhalla2008finding}, using local flow preservation has not shown noticeable speedup for the benchmark circuits in \cref{subsec:results}. Its utility may increase for larger circuits. While we postulate that it is much more likely for the change in flow to be localised, transformations which globally alter flow could be crucial as they may not have a direct analogue in the circuit formalism. Hence, this paper opts for recalculating cflow after each transformation.

\subsection{Flow upper bound}\label{subsec:cflow-bound}
For this section, we define $n \:= |V(G)|$, $m \:= |E(G)|$ and $k \:= |I|$. An important bound on the number of edges in an open graph admitting cflow is established as follows:
\begin{lemma}[\citenum{de2007extremal}]\label{lem:flow-upper-bound}
    For an open graph $(G,I,O)$ with $|I|=|O|$ that admits a cflow, the number of edges is bounded by $m \le kn-{\binom{k+1}{2}}$.
\end{lemma}
Consider a graph where the existence of cflow is preserved under \LCompRule transformations without neighbour unfusions. Focus on the open subgraph of the transformation vertex set and its image, $(G_\mT, I_\mT, O_\mT) \smallto (G'_\mT, I'_\mT, O'_\mT)$, for the original and new successor functions, respectively. In general this could be a global change, hence input and output sets are allowed to change. Analogous to the calculation of heuristics in \refcite{staudacher2022reducing}, one can establish the following relations:
\begin{align}
    m \le kn - \frac{1}{2}k(k+1) & \quad & m' = \frac{1}{2}n(n-1) - m  \le k'(n-1) - \frac{1}{2}k'(k'+1)
\end{align}
Rearranging derives the following condition for $m$:
\begin{equation}
    \frac{1}{2}(n-k')(n-k'-1) \le m \le kn - \frac{1}{2}k(k+1)
\end{equation}
This depends on the number of inputs to each subgraph and can be used to prune local flow preserving transformations. In general, using bounds $1 \le k \le n-2$, $1 \le k' \le n-1$ simplifies this condition to the trivial $2 \le m \le \frac{1}{2}n(n-1) - 1$, assuming at least two edges. A similar analysis applies to \LCompRule with \NeighbourUnfuseRule as well as \PivotRule. However, these general conditions are not very restrictive for broader cases.
\section{Optimisation of quantum circuits}\label{sec:optimisation}
The optimisation strategy begins with translating the circuit into a ZX-diagram (\cref{sec:zx-calculus}) and applying phase teleportation to maximise T-count reduction. Phase jumping increases the space of potential transformations for reducing 2Q-count, but a greedy approach surprisingly hinders performance. Optimal reductions were observed to result from fixing phases on the variables remaining in reduced gadget form. Our focus now shifts to 2Q-count minimisation.

The diagram is then converted into graph-like form (\cref{lem:graph-like-trans}), which admits cflow as per \refcite{duncan2020graph}. Our simplification rule set includes those defined in \cref{sec:simplify-zx}, but in general, any transformations preserving the diagram's graph-like nature could be used. While \refcite{duncan2020graph} reduces vertex count via the iterative application of \IdFuseRule, \LCompRule and \PivotRule, this does not effectively decrease 2Q-count. Instead, a heuristic-based strategy akin to \refcite{staudacher2022reducing} is adopted, scoring each possible transformation. Asserting that transformations preserve cflow means their 2Q-count impact is precisely calculable via \cref{lem:N2Q}:
\begin{equation}\label{eq:DeltaN2Q}
    \Delta N_{2Q} = \Delta \abs{E} - \Delta \abs{V}
\end{equation}
Here, $\Delta \abs{E}$ and $\Delta \abs{V}$ for each transformation are determined analagously to the heuristics in \refcite{staudacher2022reducing}. While transformations increasing 2Q-count could lead to a more simplified final circuit, a terminating procedure is prioritised. Any which either increase 2Q-count or maintain it without removing vertices are rejected ($\Delta N_{2Q} < 0$ or $\Delta N_{2Q} = 0$ and $\Delta \abs{V} < 0$).

We define a match's score as $-\Delta N_{2Q}$, though more advanced strategies could be developed. Our decision strategy favours transformations with the largest 2Q-count reduction. After each transformation, the match list is updated, removing potentially affected matches and rechecking only the transformed part of the graph. Once termination is reached and a circuit is extracted, it undergoes a basic gate cancellation and commutation algorithm (\BasicOptimize in PyZX), eliminating redundancies. The complete algorithm's pseudocode is presented in \cref{app:pseudocode}.

\subsection{Neighbour unfusion matches}\label{subsec:nu-matches}
The algorithm identifies all matches of \IdFuseRule, \LCompRule and \PivotRule, calculating $\Delta N_{2Q}$ for each, pruning matches that could prevent termination. To ensure optimal reductions, it considers all prior applications of \NeighbourUnfuseRule. For instance, \NeighbourUnfuseRule preceding \LCompRule, can be applied on any vertex provided adjacent boundary vertices are in the unfused neighbour subset $S_N$.

\begin{equation}
    \begin{tikzpicture}
	\begin{pgfonlayer}{nodelayer}
		\node [style=Z phase dot] (0) at (2.5, 1) {$\ \pm\frac\pi2\ $};
		\node [style=Z phase dot] (2) at (0.5, 0.5) {$\beta_1$};
		\node [style=Z phase dot] (4) at (4.5, 0.5) {$\beta_m$};
		\node [style=none] (5) at (-0.5, -0.5) {};
		\node [style=none] (6) at (0.5, -0.5) {};
		\node [style=none] (7) at (4.5, -0.5) {};
		\node [style=none] (8) at (5.5, -0.5) {};
		\node [style=none] (9) at (2.5, 0) {...};
		\node [style=none] (10) at (0.05, -0.5) {...};
		\node [style=none] (11) at (4.95, -0.5) {...};
		\node [style=none] (16) at (6.75, 1) {$=$};
		\node [style=none] (17) at (14.25, -0.5) {};
		\node [style=none] (18) at (8, -0.5) {};
		\node [style=none] (19) at (13.7, -0.5) {...};
		\node [style=none] (20) at (13.25, -0.5) {};
		\node [style=Z phase dot] (24) at (9, 0.5) {$\ \beta_1\!\mp\!\frac\pi2\ $};
		\node [style=none] (26) at (9, -0.5) {};
		\node [style=none] (27) at (8.55, -0.5) {...};
		\node [style=Z phase dot] (30) at (13.25, 0.5) {$\ \beta_m \!\mp\! \frac\pi2\ $};
		\node [style=none] (31) at (1.25, -2) {};
		\node [style=Z phase dot] (32) at (1.25, -1) {$\beta_2$};
		\node [style=none] (33) at (0.8, -2) {...};
		\node [style=none] (34) at (0.25, -2) {};
		\node [style=Z phase dot] (35) at (3.75, -1) {$\ \beta_{m\!-\!1}\ $};
		\node [style=none] (36) at (3.75, -2) {};
		\node [style=none] (37) at (4.75, -2) {};
		\node [style=none] (38) at (4.2, -2) {...};
		\node [style=Z phase dot] (39) at (9.95, -1) {$\ \beta_2\!\mp\!\frac\pi2\ $};
		\node [style=none] (40) at (8.8, -2) {...};
		\node [style=none] (41) at (8.25, -2) {};
		\node [style=none] (42) at (9.25, -2) {};
		\node [style=Z phase dot] (43) at (12.55, -1) {$\ \beta_{m\!-\!1} \!\mp\! \frac\pi2\ $};
		\node [style=none] (44) at (14, -2) {};
		\node [style=none] (45) at (13.45, -2) {...};
		\node [style=none] (46) at (13, -2) {};
		\node [style=none] (47) at (11.125, 0) {...};
		\node [style=none] (48) at (3.25, 1.5) {\color{red} $*$};
		\node [style=none] (49) at (6.75, 1.75) {\lcomprule};
		\node [style=none] (54) at (1, 5.75) {};
		\node [style=none] (55) at (2, 5.75) {};
		\node [style=none] (56) at (3, 5.75) {};
		\node [style=none] (57) at (4, 5.75) {};
		\node [style=Z phase dot] (58) at (1.5, 4.75) {$\gamma_1$};
		\node [style=Z phase dot] (59) at (3.5, 4.75) {$\gamma_n$};
		\node [style=Z phase dot] (60) at (2.5, 3.5) {$\ \alpha \mp \frac\pi2\ $};
		\node [style=Z dot] (61) at (2.5, 2.25) {};
		\node [style=none] (62) at (1.5, 5.75) {...};
		\node [style=none] (63) at (3.5, 5.75) {...};
		\node [style=none] (64) at (2.5, 4.25) {...};
		\node [style=none] (65) at (9.625, 5.75) {};
		\node [style=none] (66) at (10.625, 5.75) {};
		\node [style=none] (67) at (11.625, 5.75) {};
		\node [style=none] (68) at (12.625, 5.75) {};
		\node [style=Z phase dot] (69) at (10.125, 4.75) {$\gamma_1$};
		\node [style=Z phase dot] (70) at (12.125, 4.75) {$\gamma_n$};
		\node [style=Z phase dot] (71) at (11.125, 3.5) {$\ \alpha \mp \frac\pi2\ $};
		\node [style=none] (73) at (10.125, 5.75) {...};
		\node [style=none] (74) at (12.125, 5.75) {...};
		\node [style=none] (75) at (11.125, 4.25) {...};
		\node [style=Z phase dot] (76) at (11.125, 2.25) {$\ \mp \frac{\pi}{2}\ $};
		\node [style=Z phase dot] (77) at (-6, 1) {$\ \alpha\ $};
		\node [style=Z phase dot] (78) at (-8, 0.5) {$\beta_1$};
		\node [style=Z phase dot] (79) at (-4, 0.5) {$\beta_m$};
		\node [style=none] (80) at (-9, -0.5) {};
		\node [style=none] (81) at (-8, -0.5) {};
		\node [style=none] (82) at (-4, -0.5) {};
		\node [style=none] (83) at (-3, -0.5) {};
		\node [style=none] (84) at (-6, 0) {...};
		\node [style=none] (85) at (-8.45, -0.5) {...};
		\node [style=none] (86) at (-3.55, -0.5) {...};
		\node [style=none] (87) at (-7.25, -2) {};
		\node [style=Z phase dot] (88) at (-7.25, -1) {$\beta_2$};
		\node [style=none] (89) at (-7.7, -2) {...};
		\node [style=none] (90) at (-8.25, -2) {};
		\node [style=Z phase dot] (91) at (-4.75, -1) {$\ \beta_{m\!-\!1}\ $};
		\node [style=none] (92) at (-4.75, -2) {};
		\node [style=none] (93) at (-3.75, -2) {};
		\node [style=none] (94) at (-4.3, -2) {...};
		\node [style=none] (96) at (-7.25, 4) {};
		\node [style=none] (97) at (-6.25, 4) {};
		\node [style=none] (98) at (-5.75, 4) {};
		\node [style=none] (99) at (-4.75, 4) {};
		\node [style=Z phase dot] (100) at (-6.75, 3) {$\gamma_1$};
		\node [style=Z phase dot] (101) at (-5.25, 3) {$\gamma_n$};
		\node [style=none] (104) at (-6.75, 4) {...};
		\node [style=none] (105) at (-5.25, 4) {...};
		\node [style=none] (106) at (-6, 2.25) {...};
		\node [style=none] (107) at (-1.75, 1) {$=$};
		\node [style=none] (108) at (-1.75, 1.75) {\neighbourunfuserule};
		\node [style=none, rotate=-90] (109) at (-6, 5) {$\Bigg\{$};
		\node [style=none] (110) at (-6, 5.75) {$S_N$};
	\end{pgfonlayer}
	\begin{pgfonlayer}{edgelayer}
		\draw (5.center) to (2);
		\draw (6.center) to (2);
		\draw (7.center) to (4);
		\draw (8.center) to (4);
		\draw [style=hadamard edge] (2) to (0);
		\draw [style=hadamard edge] (4) to (0);
		\draw (18.center) to (24);
		\draw (26.center) to (24);
		\draw (20.center) to (30);
		\draw (17.center) to (30);
		\draw (34.center) to (32);
		\draw (31.center) to (32);
		\draw (36.center) to (35);
		\draw (37.center) to (35);
		\draw (41.center) to (39);
		\draw (42.center) to (39);
		\draw (46.center) to (43);
		\draw (44.center) to (43);
		\draw [style=hadamard edge] (24) to (30);
		\draw [style=hadamard edge] (30) to (43);
		\draw [style=hadamard edge] (43) to (39);
		\draw [style=hadamard edge] (39) to (24);
		\draw [style=hadamard edge] (24) to (43);
		\draw [style=hadamard edge] (39) to (30);
		\draw [style=hadamard edge] (0) to (32);
		\draw [style=hadamard edge] (0) to (35);
		\draw (54.center) to (58);
		\draw (55.center) to (58);
		\draw (56.center) to (59);
		\draw (57.center) to (59);
		\draw [style=hadamard edge] (61) to (0);
		\draw [style=hadamard edge] (60) to (61);
		\draw [style=hadamard edge] (60) to (59);
		\draw [style=hadamard edge] (58) to (60);
		\draw (65.center) to (69);
		\draw (66.center) to (69);
		\draw (67.center) to (70);
		\draw (68.center) to (70);
		\draw [style=hadamard edge] (71) to (70);
		\draw [style=hadamard edge] (69) to (71);
		\draw [style=hadamard edge] (71) to (76);
		\draw [style=hadamard edge] (24) to (76);
		\draw [style=hadamard edge] (76) to (39);
		\draw [style=hadamard edge] (76) to (30);
		\draw [style=hadamard edge] (76) to (43);
		\draw (80.center) to (78);
		\draw (81.center) to (78);
		\draw (82.center) to (79);
		\draw (83.center) to (79);
		\draw [style=hadamard edge] (78) to (77);
		\draw [style=hadamard edge] (79) to (77);
		\draw (90.center) to (88);
		\draw (87.center) to (88);
		\draw (92.center) to (91);
		\draw (93.center) to (91);
		\draw [style=hadamard edge] (77) to (88);
		\draw [style=hadamard edge] (77) to (91);
		\draw (96.center) to (100);
		\draw (97.center) to (100);
		\draw (98.center) to (101);
		\draw (99.center) to (101);
		\draw [style=hadamard edge] (100) to (77);
		\draw [style=hadamard edge] (77) to (101);
	\end{pgfonlayer}
\end{tikzpicture}
\end{equation}\vspace{0pt}

\noindent Selecting $S_N$ can effectively isolate a highly connected subset for complementation, maximising 2Q-count reduction. However, optimally choosing $S_N$ reduces to solving the NP-complete max cut problem on a fully connected graph with edge weights of -1 for connected neighbours, and +1 otherwise. Similar logic applies to \PivotRule, involving subset selections for both vertices.

Due to exponential growth in matches with more neighbours, checking all subsets becomes inefficient. We limit the search by setting a maximum subset size $s_{max}$ per vertex. For small $s_{max}$, subset numbers scale approximately polynomially. Our implementation checks all subsets with $|S_N| \le s_{max}$. For \PivotRule, it considers all subset combinations for both vertices.

\cref{fig:nu-smax-heatmap} demonstrates how varying $s_{max}$ for \LCompRule and \PivotRule affects total optimisation time and average 2Q-count reduction, averaged over a series of randomly generated circuits. The heatmaps indicate that increasing $s_{max}$ \PivotRule up to 2 significantly enhances effectiveness relative to additional runtime. Beyond 2, the benefits plateau against a steep rise in optimisation time. For \LCompRule, changes in $s_{max}$ impact both time and 2Q-count reduction less significantly.

\vspace{5pt}\begin{figure}[htb]
    \centering
    \includegraphics[width=0.8\linewidth]{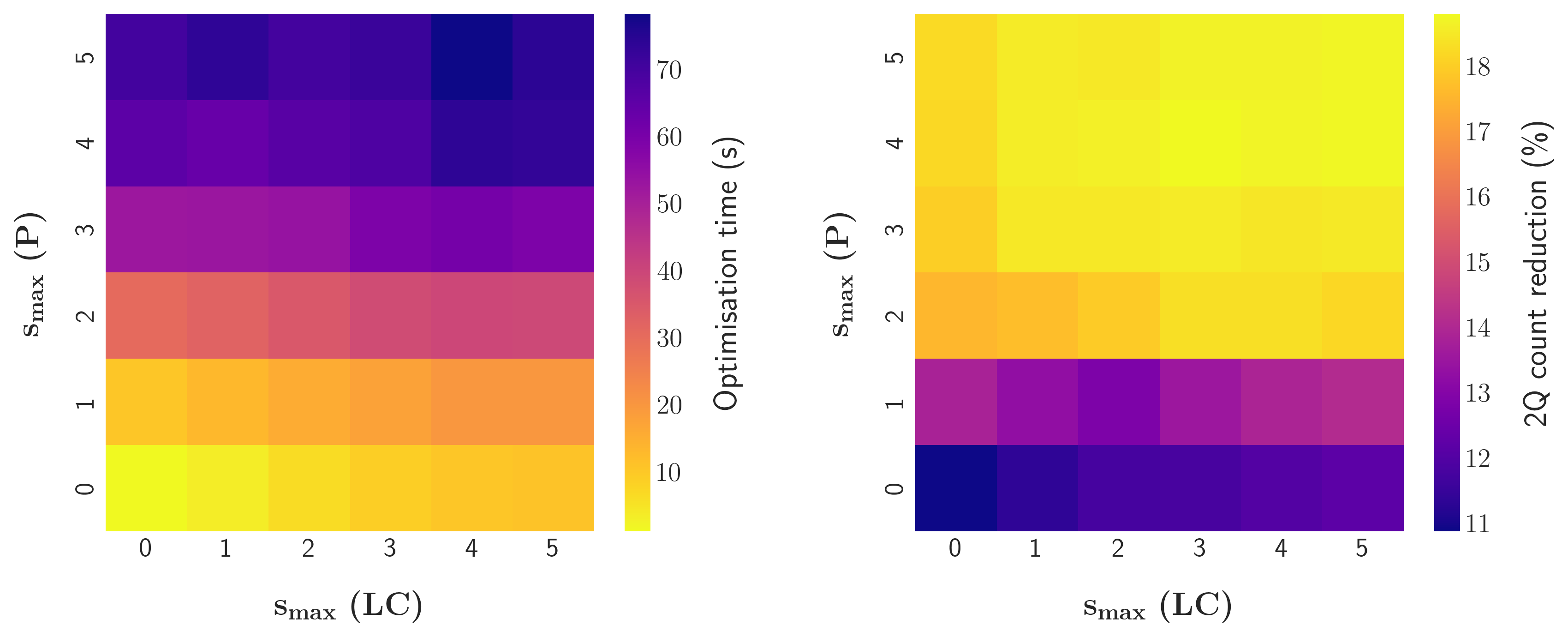}
    \caption{Heatmaps showing optimisation time and percentage reduction in 2Q-count as a function of maximum neighbour subset size for unfusions preceding local complementation \LCompRule and pivoting \PivotRule, averaged over randomly generated circuits.}\label{fig:nu-smax-heatmap}
\end{figure}

\noindent Using the bound derived in \cref{subsec:cflow-bound}, one can analyse the range of $\Delta N_{2Q} = (m'-m)-(n'-n)$ for \LCompRule. Rearranging obtains $m = \frac{1}{4}n(n-1)+\frac{1}{2}(1-\Delta N_{2Q})$. Plugging this into $2 \le m \le \frac{1}{2}n(n-1) - 1$ and solving derives the following relations between $n$ and $\Delta N_{2Q}$:
\begin{align*}
    n^2 - n - 6 - 2 \Delta N_{2Q} \ge 0 & \quad & n^2 - n - 3 + \Delta N_{2Q} \ge 0
\end{align*}
Hence, for cflow preservation, $-3 \le \Delta N_{2Q} \le 3$. This condition is general, but stricter bounds on the number of inputs could be applied, rather than using the trivial result.

Applying identical analysis to \LCompRule with \NeighbourUnfuseRule we find $m'=\frac{1}{2}(n-1)(n+2) - m + 1$. Note that $S_N$ is not included in the subgraph for this calculation as the edges of such vertices have not been modified (if we let the original vertex be the unfused one). Applying \cref{lem:flow-upper-bound} and the trivial input bounds $1 \le k \le n$, $1 \le k' \le n-1$ results in $2 \le m \le \frac{1}{2}n(n-1)$. Relating this to $\Delta N_{2Q}$ results in the following relations:
\begin{align*}
    n^2 + n - 10 - 2 \Delta N_{2Q} \ge 0 & \quad & n^2 - 3n + 2 + \Delta N_{2Q} \ge 0
\end{align*}
Solving arrives at the inequality $-5 \le \Delta N_{2Q} \le 0$. Thus local complementations with neighbour unfusions which preserve the existence of cflow never increase the 2Q-count. Similar analysis applies to \PivotRule, though the increased complexity means it does not output such neat restrictions.

\vspace{5pt}\subsection{Benchmarking results}\label{subsec:results}
Circuit metrics for various optimisation strategies are compared in \cref{fig:gen-results} and \cref{tab:benchmark-results}. \refcite{nam2018automated} (\NRSCM) reports the best 2Q-count using state-of-the-art non-ZX-calculus routines. The full strategy of \refcite{kissinger2019reducing} (\FullReduce) simplifies ZX-diagrams to reduced gadget form before circuit extraction, implemented in PyZX. \refcite{staudacher2022reducing} (\ZXHeur) employs heuristic-based routines with various decision strategies, reporting the best 2Q-count. The `greedy' \fs{g}{0} and `greedy with neighbour unfusion' \fs{g}{1} strategies are shown here as they resemble the algorithm of this paper. Alternative strategies yielding a lower 2Q-count are also shown in \cref{tab:benchmark-results} (2Q*). \FlowOpt, the primary focus of this paper, is tested in various configurations of flow preservation and neighbour unfusion limits (\textit{flow}$^{\textit{s}_{\textit{max}}}$).

\begin{figure}[H]
    \centering
    \caption{Comparison of different circuit optimisation strategies for random 8-qubit circuits, varying the proportion of T-gates $P_t$. A range of strategies are shown, labelled by the key \textit{flow}$^{\textit{s}_{\textit{max}}}$, where \textit{flow} indicates the flow property preserved (cflow or gflow), and \textit{s}$_{\textit{max}}$ represents the maximum unfused nieghbours per vertex.}\label{fig:gen-results}
    \includegraphics[width=0.95\linewidth]{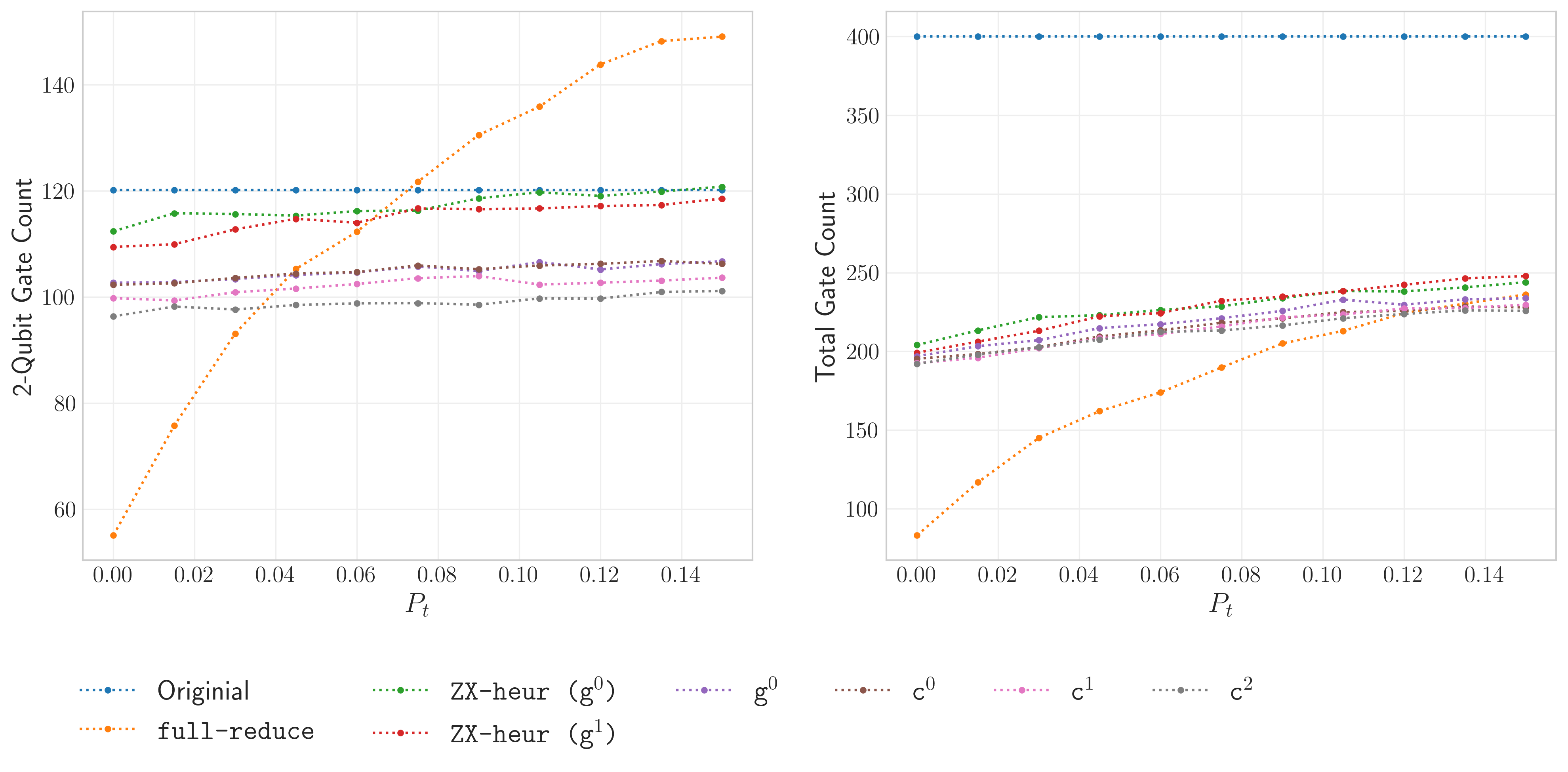}
    \par\vspace{5pt}
    \includegraphics[width=0.95\linewidth]{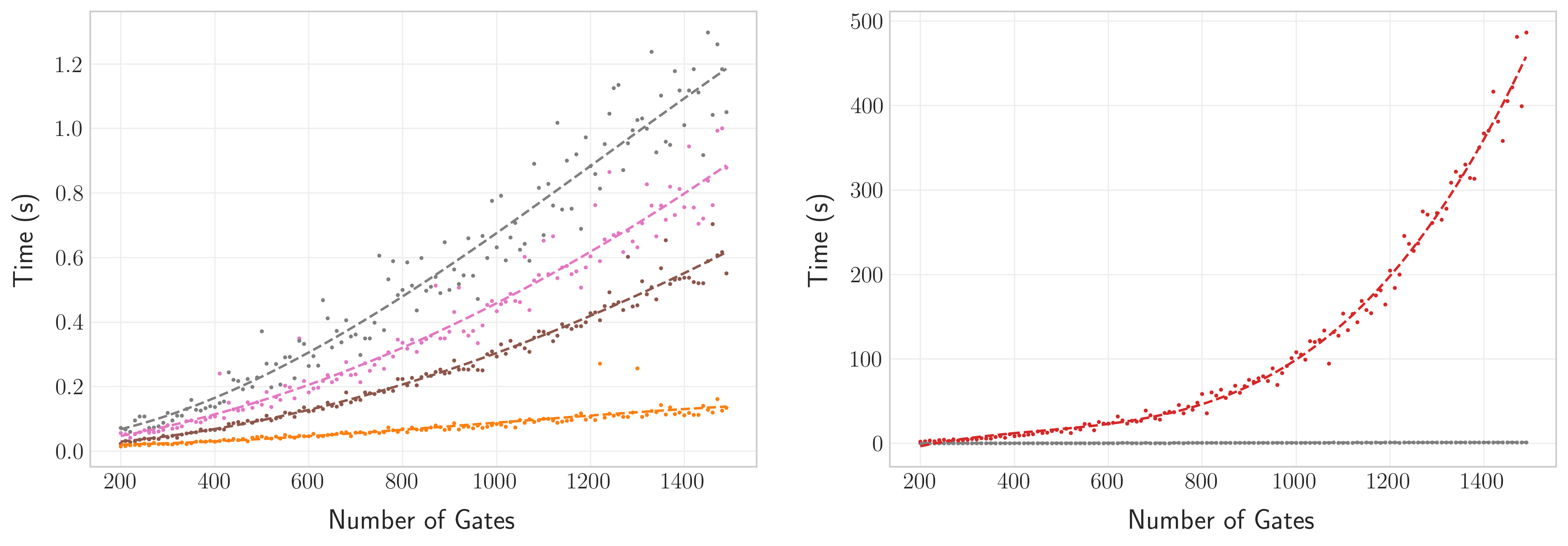}
\end{figure}

\noindent Data for \cref{fig:gen-results} (top) comes from generating random 8-qubit Clifford+T circuits with 400 gates, varying T-gate probability $P_t$ from 0 to 0.15. For each value we generated 20 circuits and reported the average metrics. \FullReduce excels for circuits with a low $P_t$. For higher $P_t$, total gate count saturates to a similar value across strategies. However, \FullReduce's 2Q-count often surpasses the original circuit's. Selective ZX-based strategies yield almost constant 2Q-count reductions regardless of $P_t$ (only marginally decreasing). All \FlowOpt strategies consistently beat \ZXHeur in 2Q-count reduction, with \FlowOpt\fs{g}{0} even outperforming \ZXHeur\fs{g}{1}. The key distinction is \ZXHeur's focus on reducing edge count, versus \FlowOpt's heuristic of \cref{eq:DeltaN2Q}. Furthermore, the implementation of \ZXHeur has redundancies not present in \FlowOpt. Notably, \FlowOpt\fs{g}{0} and \FlowOpt\fs{c}{0} show similar performance on the circuit set, indicating that extra simplifications and increases in 2Q-count during extraction balance each other. Increasing \textit{s}$_{\textit{max}}$ improves performance, with \FlowOpt\fs{c}{2} showing the best results (higher $\textit{s}_{\textit{max}}$ not shown).

\cref{fig:gen-results} (bottom) shows total optimisation time variation with circuit size. \texttt{full-reduce} scales best, limited only by extraction, as it iteratively applies matches without checks. \FlowOpt\texttt{(c)} scales slightly worse, largely due to match list updates post-application. Cflow preservation checks increase runtime but not scaling. As $\textit{s}_{\textit{max}}$ grows, scaling worsens due to more \NeighbourUnfuseRule matches, yet remains faster than \ZXHeur\fs{g}{1}, constrained by the gflow algorithm, as well as most other circuit-based strategies.

\newpage
\begin{table}[ht]
    \centering
    \footnotesize
    \renewcommand{\arraystretch}{0.9}
    \caption{Qubit count (Q), 2Q-count (2Q) and T-count (T) for original benchmark circuits \cite{circuits}, optimisation routines of \refcite{nam2018automated} (\NRSCM), heuristic-based simplification methods of \refcite{staudacher2022reducing} (\ZXHeur) and our flow preserving algorithm (\FlowOpt). In \ZXHeur the `greedy with neighbour unfusion' decision strategy is shown in the column (2Q), while lower counts from other strategies are in the column (2Q*). Various \FlowOpt strategies are shown, labeled by \textit{flow}$^{\textit{s}_{\textit{max}}}$. Here, \textit{flow} indicates the preserved property (cflow or gflow), and \textit{s}$_{\textit{max}}$ is the limit on neighbours to be unfused per vertex. The best 2Q-count for each circuit is highlighted. The algorithms' average 2Q-count reductions are shown, with values from (2Q) used for (2Q*) when empty.}\label{tab:benchmark-results}
    \begin{tabular}{l|C{\cw}C{\cw}C{\cw}|C{\cw}C{\cw}|C{\cw}C{\cw}C{\cw}|C{\cw}C{\cw}C{\cw}C{\cw}C{\cw}C{\cw}C{\cw}C{\cw}C{\cw}}
        \toprule
        \multirow{3}[10]{\fw}{\bf{Circuit}}                                                     &
        \multicolumn{3}{c|}{\multirow{2}{*}{\small \texttt{\textbf{Original}}}}                 &
        \multicolumn{2}{c|}{\multirow{2}{*}{\small \texttt{\textbf{NRSCM}}}}                    &
        \multicolumn{3}{c|}{\multirow{2}{*}{\small \texttt{\textbf{ZX-heur (g$^\mathbf{1}$)}}}} &
        \multicolumn{8}{c}{\small \texttt{\textbf{flow-opt}}}                                                                                                                                                                                                                                                                     \\
        {}                                                                                      &
        \multicolumn{3}{c|}{}                                                                   &
        \multicolumn{2}{c|}{}                                                                   &
        \multicolumn{3}{c|}{}                                                                   &
        \texttt{\textbf{g$^\mathbf{0}$}}                                                        &
        \texttt{\textbf{c$^\mathbf{0}$}}                                                        &
        \texttt{\textbf{c$^\mathbf{1}$}}                                                        &
        \texttt{\textbf{c$^\mathbf{2}$}}                                                        &
        \texttt{\textbf{c$^\mathbf{3}$}}                                                        &
        \texttt{\textbf{c$^\mathbf{4}$}}                                                        &
        \texttt{\textbf{c$^\mathbf{5}$}}                                                        &                                                                                                                                                                                                                                 \\ \cmidrule{2-17}

        {}                                                                                      & \textbf{Q} & \textbf{2Q} & \textbf{T} & \textbf{2Q} & \textbf{T} & \textbf{2Q} & \textbf{2Q*} & \textbf{T} & \multicolumn{7}{c}{\textbf{2Q}} & \textbf{T}                                                                       \\ \midrule
        adder$_8$                                                                               & 24         & 409         & 399        & 291         & 215        & 295         & -            & 173        & 296                             & 295        & 284       & 277        & 269        & \best{267} & 268        & 173 \\
        barenco-tof$_4$                                                                         & 7          & 48          & 56         & \best{34}   & 28         & 41          & 40           & 28         & 42                              & 42         & 37        & 37         & 37         & 37         & 37         & 28  \\
        barenco-tof$_5$                                                                         & 9          & 72          & 84         & \best{50}   & 40         & 60          & 57           & 40         & 63                              & 63         & 57        & 55         & 55         & 55         & 55         & 40  \\
        barenco-tof$_{10}$                                                                      & 19         & 192         & 224        & \best{130}  & 100        & 162         & 151          & 100        & 159                             & 159        & 151       & 146        & 146        & 146        & 146        & 100 \\
        tof$_4$                                                                                 & 7          & 30          & 35         & \best{22}   & 23         & 24          & -            & 23         & 24                              & 24         & 24        & 24         & 24         & 24         & 24         & 23  \\
        tof$_5$                                                                                 & 9          & 42          & 49         & \best{30}   & 31         & 33          & -            & 31         & 33                              & 33         & 33        & 33         & 33         & 33         & 33         & 31  \\
        tof$_{10}$                                                                              & 19         & 102         & 119        & \best{70}   & 71         & 78          & -            & 71         & 78                              & 78         & 78        & 78         & 78         & 78         & 78         & 71  \\
        csla-mux$_3$                                                                            & 15         & 80          & 70         & \best{70}   & 64         & 72          & -            & 62         & 74                              & 74         & 73        & 73         & 73         & 73         & 73         & 62  \\
        csum-mux$_9$                                                                            & 30         & 168         & 196        & \best{140}  & 84         & 156         & 150          & 84         & 151                             & 152        & 150       & \best{140} & \best{140} & \best{140} & \best{140} & 84  \\
        gf($2^4$)-mult                                                                          & 12         & 99          & 112        & 99          & 68         & 101         & -            & 68         & 98                              & 99         & 99        & \best{94}  & \best{94}  & \best{94}  & \best{94}  & 68  \\
        gf($2^5$)-mult                                                                          & 15         & 154         & 175        & 154         & 115        & 156         & -            & 115        & 153                             & 154        & 154       & \best{146} & \best{146} & \best{146} & \best{146} & 115 \\
        gf($2^6$)-mult                                                                          & 18         & 221         & 252        & 221         & 150        & 223         & 217          & 150        & 217                             & 221        & 221       & \best{209} & \best{209} & \best{209} & \best{209} & 150 \\
        gf($2^7$)-mult                                                                          & 21         & 300         & 343        & 300         & 217        & 299         & -            & 217        & 293                             & 300        & 300       & \best{283} & \best{283} & \best{283} & \best{283} & 217 \\
        gf($2^8$)-mult                                                                          & 24         & 405         & 448        & 405         & 264        & 411         & 405          & 264        & 395                             & 405        & 405       & \best{383} & \best{383} & \best{383} & \best{383} & 264 \\
        mod-mult-55                                                                             & 9          & 48          & 49         & \best{40}   & 35         & \best{40}   & -            & 35         & \best{40}                       & \best{40}  & \best{40} & \best{40}  & \best{40}  & \best{40}  & \best{40}  & 35  \\
        mod-red-21                                                                              & 11         & 105         & 119        & \best{77}   & 73         & 87          & 85           & 73         & 85                              & 87         & 86        & 83         & 83         & 83         & 83         & 73  \\
        mod5$_4$                                                                                & 5          & 28          & 28         & 28          & 16         & 23          & -            & 8          & 23                              & 25         & 23        & \best{21}  & \best{21}  & \best{21}  & \best{21}  & 8   \\
        qcla-adder$_{10}$                                                                       & 36         & 233         & 238        & 183         & 162        & 201         & 193          & 162        & 200                             & 200        & 189       & 182        & 180        & \best{174} & 175        & 162 \\
        qcla-com$_7$                                                                            & 24         & 186         & 203        & 132         & 95         & 140         & 138          & 95         & 136                             & 136        & 134       & 133        & 133        & \best{131} & \best{131} & 95  \\
        qcla-mod$_7$                                                                            & 26         & 382         & 413        & \best{292}  & 235        & 313         & 311          & 237        & 312                             & 312        & 310       & 296        & 293        & 293        & \best{292} & 237 \\
        rc-adder$_6$                                                                            & 14         & 93          & 77         & \best{71}   & 47         & \best{71}   & -            & 47         & \best{71}                       & \best{71}  & \best{71} & \best{71}  & \best{71}  & \best{71}  & \best{71}  & 47  \\
        vbe-adder$_3$                                                                           & 10         & 70          & 70         & 50          & 24         & 44          & 42           & 24         & 46                              & 46         & 44        & 39         & 40         & \best{36}  & \best{36}  & 24  \\\midrule \midrule
        $\langle \text{2Q}\downarrow \rangle \;\; (\%)$                                         &            &             &            & 18.47       &            & 14.57       & 15.91        &            & 15.07                           & 14.27      & 16.33     & 19.24      & 19.34      & 19.79      & 19.77      &     \\\midrule
        \bottomrule
    \end{tabular}
\end{table}\vspace{10pt}
\noindent \cref{tab:benchmark-results} compares circuit metrics for a standard set of benchmark circuits \cite{amy2014polynomial}, including key components of quantum algorithms like adders and oracles. The table highlights the effectiveness of preserving cflow in the \FlowOpt algorithm, outperforming both ZX-based strategies (\ZXHeur) and state-of-the-art non-ZX strategies (\NRSCM). Phase teleportation aligns the T-counts with the best strategies for almost all circuits. While lower T-counts have been found by combining phase teleportation with the TODD algorithm of \refcite{heyfron2018efficient} \cite{kissinger2019reducing}, it is slower and significantly increases the 2Q-count.

The table reveals that deviating from a circuit-like structure in \FlowOpt\fs{g}{0} enhances performance over \FlowOpt\fs{c}{0}, averaging a 15.07\% reduction compared to 14.27\%. The most substantial improvement is seen with increasing $\textit{s}_{\textit{max}}$, demonstrated by \FlowOpt\fs{c}{2}, achieving a 19.24\% reduction. This gain plateaus for $\textit{s}_{\textit{max}} > 2$, showing minor improvements at the cost of longer runtime, as evidenced in \cref{fig:nu-smax-heatmap}. On this set of circuits increasing $\textit{s}_{\textit{max}}$ from \FlowOpt\fs{c}{4} to \FlowOpt\fs{c}{5} in fact marginally worsens performance. \FlowOpt\fs{c}{4} excels with a 19.79\% average reduction, surpassing \NRSCM's previous best of 18.47\%. For many circuits new 2Q-count minimums are found. Notably, \FlowOpt\fs{c}{2} optimises all five gf-mult circuits, outperforming \NRSCM, which shows no improvement in these cases. Compared to standard Clifford simplifications from \refcite{duncan2020graph}, both strategies show marked improvements. \FlowOpt's success over \NRSCM, which aggregates results from a variety of routines and includes several hundred rewrite rules, underscores its efficiency.

\subsection{QFT circuit optimisation}
In the case of QFT circuits, both \NRSCM and \FlowOpt\fs{c}{2} fail to reduce the 2Q-count. However, \FlowOpt\fs{g}{0} optimises each circuit such that \emph{the 2Q-count is equal to the non-Clifford gate count}, as shown in \cref{tab:qft-benchmark-results}. This optimisation seems to achieve a regular structure, shown in \cref{fig:qft8} for QFT$_8$.
\begin{SCtable}[\sidecaptionrelwidth][ht]
    \centering
    \footnotesize
    \renewcommand{\arraystretch}{0.9}
    \begin{tabular}{l|C{\cw}C{\cw}C{\cw}|C{\cw}C{\cw}|C{\cw}C{\cw}C{\cw}|C{\cw}C{\cw}}
        \toprule
        \multirow{3}[10]{\fw}{\bf{Circuit}}                                     &
        \multicolumn{3}{c|}{\multirow{2}{*}{\small \texttt{\textbf{Original}}}} &
        \multicolumn{2}{c|}{\multirow{2}{*}{\small \texttt{\textbf{NRSCM}}}}    &
        \multicolumn{3}{c|}{\small \texttt{\textbf{flow-opt}}}                  &
        \multicolumn{2}{c}{\multirow{2}{*}{\small \texttt{\textbf{qft-opt}}}}                                                                                                                                                                 \\
        {}                                                                      &
        \multicolumn{3}{c|}{}                                                   &
        \multicolumn{2}{c|}{}                                                   &
        \texttt{\textbf{g$^\mathbf{0}$}}                                        &
        \texttt{\textbf{c$^\mathbf{2}$}}                                        & \multicolumn{1}{c|}{}                                                                                                                                       \\ \cmidrule{2-11}
        {}                                                                      & \textbf{Q}            & \textbf{2Q} & \textbf{T} & \textbf{2Q} & \textbf{T} & \multicolumn{2}{c}{\textbf{2Q}} & \textbf{T} & \textbf{2Q} & \textbf{T}       \\ \midrule
        QFT$_8$                                                                 & 8                     & 56          & 84         & 56          & 42         & \best{42}                       & 56         & 42          & \best{42}  & 42  \\
        QFT$_{16}$                                                              & 16                    & 228         & 342        & 228         & 144        & \best{144}                      & 228        & 144         & \best{144} & 144 \\
        QFT$_{32}$                                                              & 32                    & 612         & 918        & 612         & 368        & \best{368}                      & 612        & 368         & \best{368} & 368 \\
        \bottomrule
    \end{tabular}
    \noindent\caption{\raggedright Qubit count (Q), 2Q-count (2Q) and non-Clifford gate count (T) for original QFT benchmark circuits \cite{circuits}, optimisation routines of  \refcite{nam2018automated} (\NRSCM), the presented flow preserving algorithm (\FlowOpt) and a minimal case algorithm (\QFTOpt).}\label{tab:qft-benchmark-results}
\end{SCtable}

\noindent A minimal algorithm (\QFTOpt), which achieves these results, comprises phase teleportation, iterative \IdFuseRule application, circuit extraction using general method from \refcite{backens2021there}, and basic circuit commutations and cancellations. It always maintains cflow, with 2Q-count reductions stemming solely from the general extraction algorithm. Using the cflow extraction approach on the same diagram does not impact the 2Q-count. This phenomenon is specific to QFT circuits; in other benchmarks, general extraction mirrors cflow extraction results. Also note that the choice of spider for fixing teleported phases is pivotal, as different choices yield 2Q-counts between this minimum and the original.

\vspace{20pt}\begin{figure}[htp]
    \centering
    \includegraphics[height=2.2cm]{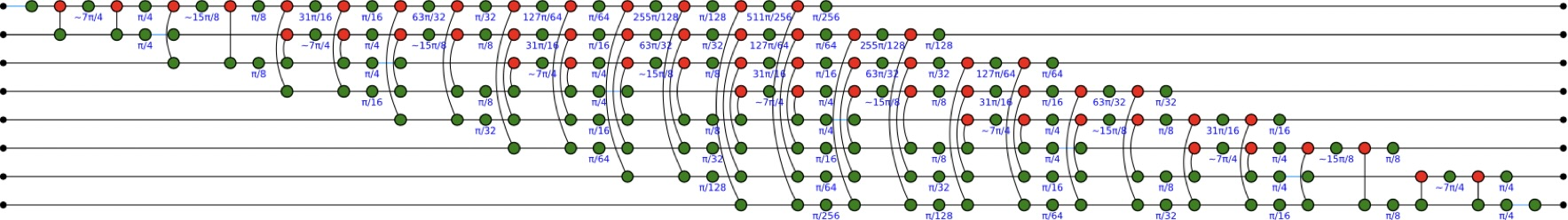}
    \par\vspace{0.8cm}
    \includegraphics[height=2.2cm]{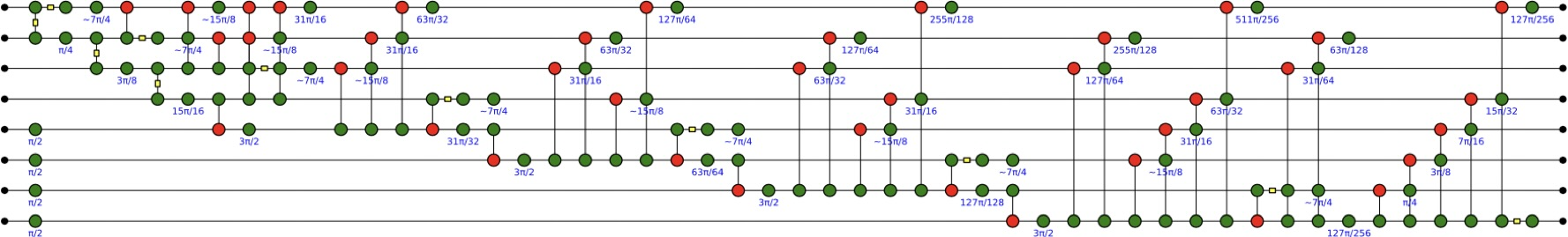}
    \caption{Benchmark circuit QFT$_8$, before (top) and after (bottom) optimisation by \QFTOpt.}\label{fig:qft8}
\end{figure}

\section{Conclusions and future work}\label{sec:conclusion}
This paper has introduced a deterministic quantum circuit optimisation strategy that surpasses current state-of-the-art algorithms. Translating circuits into ZX-diagrams enabled the use of graph-theoretic transformations, valid in the ZX-calculus, for simplification. The application of phase teleportation effectively reduces the T-count without altering the circuit's structure. Preserving causal flow then increases the mutual information between a transformation's immediate effect on the diagram and the resulting circuit's metrics. This enables the ranking of transformations based on their efficacy in reducing 2Q-count, allowing for the greedy application of the most effective transformations. Simplification rules of local complementation, pivoting and a renamed identity fusion rule are used. Additionally, the neighbour unfusion rule is generalised, showing that allowing multiple neighbours to be unfused proves to be particularly effective for 2Q-count reduction. This work opens multiple avenues for future research:
\begin{itemize}
    \item Using the equation for a transformation's effect on 2Q-count as an approximate heuristic for gflow preservation is more effective than preserving cflow when neighbour unfusions were not allowed. However, since neighbour unfusion does not guaranteed gflow preservation, the inefficiency of the gflow algorithm means that full recalculation is infeasible. Identifying specific conditions under which neighbour unfusion preserves gflow could yield superior results, albeit at the cost of losing direct visibility of the extracted circuit during simplification.
    \item Expanding the range of simplification rules would be advantageous. Notably, unlike \FullReduce \cite{kissinger2019reducing}, our approach does not necessitate spider removal in transformations. Options such as chaining multiple transformations are viable. Furthermore, given the defined qubit lines in diagrams, it is feasible to adapt known circuit rewrite rules into ZX-diagram rewrite rules which are conditional on the flow. As demonstrated in \refcite{kissinger2019reducing} and \refcite{staudacher2022reducing}, combining ZX-based and non-ZX approaches often yields the best optimisation. This could be a method of hybridisation between them.
    \item Several improvements could boost efficiency whilst maintaining approximate effectiveness. For example, streamlining the selection and verification of neighbour unfusion matches would significantly accelerate the algorithm. Pruning some matches without full checks, or employing modified versions of algorithms for the max cut problem, would contribute to faster processing. Moreover, establishing criteria for certain transformations to preserve flow without requiring checks would also hasten the process.
    \item Another area for future work is leveraging cflow's direct visibility of qubit connectivity, which could enable adaptation of the algorithm to constrained qubit topologies.
    \item Exploring more sophisticated decision strategies is also promising. Incorporating additional transformation characteristics, (e.g. non-Clifford spider density) into the ranking system could enhance optimisation. Transformations that initially increase 2Q-count might enable subsequent transformations for greater reductions. Applications for machine learning algorithms, such as simulated annealing and reinforcement learning, are considerable here. It was also observed that certain circuit classes often achieve maximum optimisation through similar strategies, suggesting that circuit classification could enhance both accuracy and general applicability.
    \item Finally, a detailed analysis of \QFTOpt's effectiveness remains to be conducted.
\end{itemize}

\vspace{20pt}\bibliography{references.bib}


\appendix

\section{Pseudocode}\label{app:pseudocode}

\begin{algorithm}[H]
    \caption{Flow preserving optimisation of quantum circuits}\label{alg:flow-opt}
    \begin{algorithmic}
        \Function{flow-opt}{\textit{circuit}}
        \State $g \gets \Call{CircuitToDiagram}{circuit}$
        \State $g \gets \Call{DiagramToGraphLike}{g}$
        \State $g \gets \Call{PhaseTeleportation}{g}$
        \State $matches \gets \Call{FindMatches}{g}$
        \ForAll{$match \in matches$}
        \State $match.score \gets \Call{ScoringFunction}{match}$
        \EndFor
        \While{$matches \neq \emptyset$}
        \State $match \gets \Call{MaxScore}{matches}$
        \State $g' \gets \Call{ApplyMatch}{g,match}$
        \If{$\Call{Flow}{g'}$}
        \State $matches \gets \Call{UpdateMatches}{matches, match, g, g'}$
        \State $g \gets g'$
        \Else
        \State $matches \gets matches \setminus match$
        \EndIf
        \EndWhile
        \State $circuit\_opt \gets \Call{ExtractCircuit}{g}$
        \State $circuit\_opt \gets \Call{BasicOptimise}{circuit\_opt}$
        \State \Return $circuit\_opt$
        \EndFunction
    \end{algorithmic}
\end{algorithm}

\section{Causal flow for labelled open graphs} \label{app:flow-extension}
Graph-like ZX-diagrams can be viewed as an extension of measurement-based quantum computing measurement patterns (which correspond to graph-states). If all spider phases are unfused (with phase gadgets being colour changed), then each vertex in the underlying open-graph has a label corresponding to measurement planes of the Bloch sphere.

For Graph-like ZX-diagrams with no phase gadgets, each spider is in the \XYm plane. If a phase gadget is introduced, this is in either the \YZm or \XZm plane, depending on the phase of the base of the gadget.

The definition of causal flow presented in this paper (\cref{def:causal-flow}) is constrained to diagrams where all spiders are in the \XYm plane. Here the definition of causal flow is extended to all measurement planes, using logic from \refcite{danos2006determinism}. This is analogous to the extension of generalised flow in \refcite{backens2021there}.
\begin{definition}
    A \emph{labelled open graph} is a tuple $\Gamma = (G,I,O, \lambda)$ where $(G,I,O)$ is an open graph, and $\lambda : \bar{O} \rightarrow \{ \XYm, \YZm, \XZm\}$ assigns a measurement plane to each non-output vertex.
\end{definition}
\begin{definition}
    A \emph{causal flow} $(f, \preceq)$ on a labelled open graph $(G,I,O,\lambda)$ consists of a \emph{successor function} $f: \bar{O} \to \bar{I}$ and a partial order $\preceq$ over $V(G)$\,, such that $\forall \; u \in \bar{O}$ and $v \in V(G)$:
    \begin{enumerate}
        \item $\lambda(u) \neq \XZm$
        \item If $\lambda(u)=\XYm$, then $f(u) \sim u$
        \item If $\lambda(u)=\YZm$, then $f(u) = u $
        \item If $f(u) \neq u$ then $u \preceq f(u)$
        \item If $v \sim f(u)$ and $v \neq u$ then $u \preceq v$
    \end{enumerate}
\end{definition}
\noindent Causal flow cannot be defined for graph-states with measurement planes in the \XZm plane. Spiders measured in the \YZm plane have self loops in the open graph's influencing digraph. One can then show that the underlying open graph without the spiders measured in the \YZm plane also admits a causal flow.
\begin{lemma}
    Let $(f,\preceq)$ be a causal flow for $(G,I,O,\lambda)$ and let $u \in \bar{O}$ with $\lambda(u) \neq \XYm$. Then $(f,\preceq)$ is a causal flow for $(G \setminus \{u\},I,O,\lambda)$
\end{lemma}
\begin{proof}
    Observe that $f(u) = u$ since $\lambda(u) \neq \XYm$. If a vertex $v \in \bar{O}$ had $f(v)=u$, then $v \preceq u$ contradicts $u \preceq N(f(u))=N(u)$. Therefore the successor function also defines a causal flow, as no edges have been added to the influencing digraph, hence no cycles can have been created.
\end{proof}
This means that given a graph-like ZX-diagram with causal flow, the underlying open graph of the diagram with its phase gadgets removed also admits a causal flow. This can be used to check if a diagram has causal flow. For a diagram with causal flow, all gadgets can be extracted as phase polynomials. This can be seen by considering that if any two vertices connected to the same gadget have a partial order relation between them, then there will be a causal cycle in the influencing digraph. Therefore each gadget must be connected to a unique set of qubits dipaths, hence can be extracted as a phase polynomial \cite{de2020architecture}.

\end{document}